\documentclass[11pt]{article}   
\usepackage{amsmath}\usepackage{hyperref}\usepackage{wasysym}
\usepackage[font=footnotesize]{caption}\usepackage[T1]{fontenc}\usepackage{latexsym}
\usepackage[english]{babel}\usepackage{cite}\usepackage{enumerate}
\usepackage{bbold}\usepackage[shortlabels]{enumitem}\usepackage[utf8]{inputenc}
\usepackage{dsfont}\usepackage{diagbox}\usepackage{amsfonts}\usepackage{mathtools}      
\usepackage{amssymb}\usepackage{euscript}\usepackage{braket}\usepackage{amssymb}
\usepackage{starfont}\usepackage{color,soul}\usepackage{braket}\usepackage{tensor}        
\usepackage{amsthm}\usepackage{graphicx}\usepackage{slashed}\usepackage{leftidx}
\usepackage{subfigure}\usepackage{bbm}\usepackage{empheq}\usepackage{color}
\usepackage{amsfonts}\usepackage{amssymb}\usepackage{graphicx}
\usepackage{amssymb,epsfig,subfigure}\usepackage{amssymb}\usepackage{comment}
\usepackage[T1]{fontenc}\usepackage{latexsym}\usepackage{cancel}
\usepackage{titlesec}\usepackage[numbers,sort&compress]{natbib}  
\usepackage{float}\usepackage[dvipsnames]{xcolor}
\usepackage[header,title,page,titletoc]{appendix}  
\usepackage[numbers,sort&compress]{natbib}  
\usepackage[paper=a4paper,margin=1in]{geometry}  
\allowdisplaybreaks[4]   
\numberwithin{equation}{section}
\makeatletter
\renewcommand\section{\@startsection {section}{1}{\z@}%
                                 {-3.5ex \@plus -1ex \@minus -.2ex}
                                   {2.3ex \@plus.2ex}%
                                   {\normalfont\large\bfseries}}
\renewcommand\subsection{\@startsection{subsection}{2}{\z@}%
                                   {-3.25ex\@plus -1ex \@minus -.2ex}%
                                     {1.5ex \@plus .2ex}%
                                     {\normalfont\bfseries}}
\renewcommand\subsubsection{\@startsection{subsubsection}{3}{\z@}%
                                   {-3.25ex\@plus -1ex \@minus -.2ex}%
                                     {1.5ex \@plus .2ex}%
                                     {\normalfont\itshape}}
\makeatother
\def\pplogo{\vbox{\kern-\headheight\kern -29pt
\halign{##&##\hfil\cr&{\ppnumber}\cr\rule{0pt}{2.5ex}&\ppdate\cr}}}
\makeatletter
\def\ps@firstpage{\ps@empty \def\@oddhead{\hss\pplogo}%
  \let\@evenhead\@oddhead 
}
\def\maketitle{\par
 \begingroup
 \def\thefootnote{\fnsymbol{footnote}}
 \def\@makefnmark{\hbox{$^{\@thefnmark}$\hss}}
 \if@twocolumn
 \twocolumn[\@maketitle]
 \else \newpage
 \global\@topnum\z@ \@maketitle \fi\thispagestyle{firstpage}\@thanks
 \endgroup
 \setcounter{footnote}{0}
 \let\maketitle\relax
 \let\@maketitle\relax
 \gdef\@thanks{}\gdef\@author{}\gdef\@title{}\let\thanks\relax}
\makeatother            
\definecolor{rossocorsa}{rgb}{0.83, 0.0, 0.0}
\definecolor{navyblue}{rgb}{0.0, 0.0, 0.5}
\hypersetup{
    colorlinks,
    citecolor=rossocorsa,
    filecolor=navyblue,
    linkcolor=navyblue,
    urlcolor=navyblue
} 

\newcommand\bea{\begin{eqnarray}}\newcommand\eea{\end{eqnarray}}
\newcommand{\be}{\begin{equation}}\newcommand{\ee}{\end{equation}}
\newcommand{\ba}{\begin{align}}\newcommand{\ea}{\end{align}}

\def\lan{\langle}
\def\ran{\rangle}

\def\a{\alpha}

\def\epsilon{\varepsilon}

\def\phi{\varphi}

\def\th{\theta}

\newtheorem{theorem}{Theorem}

\newtheorem{proposition}[theorem]{Proposition}

\def\2#1{{\cal #1}}

\begin{document} 

\begin{titlepage}

\begin{center}

\phantom{ }
\vspace{1cm}
{\bf \Large{Modular invariance as completeness}}
\vskip 1cm
Valentin Benedetti${}^{a}$, Horacio Casini${}^{b}$,  Yasuyuki Kawahigashi${}^{c}$, \\ Roberto Longo${}^{d}$, Javier M. Mag\'an${}^{e}$
\vskip 0.05in
\small{ ${}^{a,b,e}$\textit{Instituto Balseiro, Centro At\'omico Bariloche}}
\vskip -.4cm
\small{\textit{ 8400-S.C. de Bariloche, R\'io Negro, Argentina}}
{\sc}\\
\vskip 0.05in
\small{${}^{c}$\textit{Department of Mathematical Sciences, the University of Tokyo}}
\vskip -.4cm
\small{\textit{Komaba, Tokyo, 153-8914, Japan}}
{\sc}\\
\vskip 0.05in
\small{${}^{d}$\textit{Dipartimento di Matematica,
Universit\`a di Roma ``Tor Vergata''}}
\vskip -.4cm
\small{\textit{Via della Ricerca Scientifica, I-00133 Roma, Italy}}
{\sc}\\

\begin{abstract}
We review the physical meaning of modular invariance for unitary conformal quantum field theories in $d=2$. For QFT models,  while $T$ invariance is necessary for locality, $S$ invariance is not mandatory. $S$ invariance is a form of completeness of the theory that has a precise meaning as Haag duality for arbitrary multi-interval regions. We present a mathematical proof as well as derive this result from a physical standpoint using Renyi entropies and the replica trick. For rational CFT's, the failure of modular invariance or Haag duality can be measured by an index, related to the quantum dimensions of the model. We show how to compute this index from the modular transformation matrices. The index also appears in a limit of the Renyi mutual informations. Cases of infinite index are briefly discussed. Part of the argument can be extended to higher dimensions, where the lack of completeness can also be diagnosed using the CFT data through the thermal partition function and measured by an index. 
\end{abstract}
\end{center}

\small{\vspace{7.0 cm}\noindent ${}^{a}$valentin.benedetti@ib.edu.ar\\
${}^{b}$horaciocasini@gmail.com\\
${}^{c}$yasuyuki@ms.u-tokyo.ac.jp \\
${}^{d}$longo@mat.uniroma2.it \\
${}^{e}$javier.magan@cab.cnea.gov.ar 
}

\end{titlepage}

\setcounter{tocdepth}{2}

{\parskip = .4\baselineskip \tableofcontents}
\newpage


\section{Introduction}

Generalized symmetries (GS) describe certain general patterns that QFT may display and which are very useful for classification purposes or in the analysis of the phase structure of the theory.
They are usually described through ideas connected with the path integral formulation in non trivial manifolds. In their manifestations on the local physics, and in a model independent, real time formulation, generalized symmetries can be more precisely described as violations of Haag duality (HDV) \cite{Casini:2020rgj} (see also similar ideas in \cite{fredenhagen1982structural}). In concrete, let ${\cal A}(R)$ be the algebra of operators generated by local field operators in some topologically non trivial causal region $R$. Let $R'$ be the set of points spatially separated from $R$, the ``causal complement'' of $R$. By locality ${\cal A}(R)$ and ${\cal A}(R')$ commute with each other. More formally,  
\be 
{\cal A}(R)\subseteq ({\cal A}(R'))'\,, \label{ff}
\ee
where ${\cal A}'$ is the commutant of the algebra ${\cal A}$, i.e. the full set of operators that commute with ${\cal A}$. When the two algebras in (\ref{ff}) coincide, it is said there is Haag duality for $R$. 
A violation of this property implies there are operators that commute with the local operators in the complement of $R$ but cannot be generated by field operators in $R$ itself. These operators are then ``non-local'' in $R$, but still commute with local field operators in the complement. Several structural features are transparent in this formulation. For example,  when there are non local operators for $R$, there must exist dual non local operators for $R'$ too.  

In this light, Haag duality for general regions (related to the absence of generalized symmetries) is seen as a completeness property, implying that the algebras generated by local operators are the maximal ones compatible with causality \cite{Review}. When the theory is not complete in this sense, there are non local operators and HDV. However, note that the notion of non local operator is relative to a certain causal region $R$.\footnote{A causal region is the domain of dependence of a $d-1$-dimensional spatial region. Then a ball here is meant as the double cone determined by a $d-1$-dimensional spatial ball. In the following we use the same terminology of referring to $R$ by a spatial surface.}  A non local operator in $R$ is in general locally generated by field operators in topologically trivial regions such as balls.\footnote{In other words, this is the case if there is Haag duality for topological trivial regions such as balls. This fails for theories with spontaneously broken global symmetries \cite{roberts1974spontaneously}, though it can be repaired by taking charged operators in the algebra. Dilatation invariant theories are not spontaneously broken \cite{Roberts1974Dilatation}.}   For example, Wilson and 't Hooft loops can be locally generated in balls containing the loops, see \cite{Casini:2020rgj}.

This brings us to the question we want to address in this paper. If non local operators are ultimately locally generated in balls, then the description of a theory in terms of field operator data does not need new input in presence of GS/HDV. In particular, for a CFT, the usual bootstrap data and constraints must still mark the consistency requirement for the construction of a valid model. In this vein, it is natural to ask: What is the imprint left by GS/HDV on the bootstrap data? This important question is rather non trivial because it implies the understanding of non-local operators in terms of local ones. 

A simple example will clarify the difficulties of the problem. Suppose we have a consistent CFT in $d>2$ determined by the set of spins and scaling dimensions of the list of primaries $(\Delta_i,s_i)$, and the OPE coefficients $f_{ijk}$. This theory may have a global internal symmetry given by a group $G$. This should be simple to check, since the primaries must come in representations of this symmetry and $f_{ijk}$ must satisfy selection rules. However, if we consider only the CFT data for the neutral primary fields, we discover that this data is also perfectly consistent within this restricted set of primaries. Unitarity bounds and crossing relations will be satisfied in the neutral part of the algebra. Then the question is the converse of the above one: how to discover if certain consistent data is the neutral part of a larger one? Equivalently, how to understand if the CFT data can be extended in consistent manner? The neutral theory is a perfectly consistent unitary theory, but it is not complete in the above sense. The theory contains non local operators for regions formed by two disconnected balls. The non local operators are charge-anticharge pairs, with each charged operator having support in one of the two balls. This clearly commutes with fields outside of the two balls, but cannot be generated by neutral operators in each of the balls themselves. In the complementary region of the two balls (which has the topology of a spherical shell in general dimensions) the non local operators correspond to twists implementing the symmetry on one of the two balls only  \cite{Casini:2019kex,Casini:2020rgj}. This HDV exactly signals the lack of completeness. However, it still remains a challenge to express this idea in terms of the CFT data.

The primary objective of this article is to consider this problem in $2d$. In this restricted scenario, we will see that the problem of understanding HDV in terms of the bootstrap data can be solved in fair completion. Indeed, the theory is modular invariant if and only if it is complete. More precisely, it is invariant under the modular transformation $S$ if and only if it is complete, since invariance under the $T$ transformation is implied by locality. Then, $S$ modular invariance is ultimately a form of completeness. As such, it is generalized to higher dimensions as Haag duality for topologically non trivial regions \cite{Review}.

We note that modular invariance is usually invoked as one of the defining properties of 2d CFT's. Indeed it is behind many important applications. However, it has been difficult to justify on conceptual grounds, and some confusion still reigns around the subject. The $S$ transformation relates the partition function  $Z[a,b]\rightarrow Z[b,a]$ on a rectangular torus with lengths $a,b$. This looks like an evident geometrical symmetry for a path integral with a rotational invariant Lagrangian, and this is the origin of this symmetry for many statistical models. As such, this idea must also hold for massive theories. However, for the Euclidean description of a unitary QFT, this is not necessarily justified. It actually fails for perfectly sound unitary CFT's such as the Virasoro net, the theory generated by the stress tensor without further primaries. The reason is that the local fields integrated on the path integral may not be the actual physical fields of the theory, and some twistings in the time direction have to be enforced to project to the physically allowed states. Another common justification is that modular invariance acts as a discrete conformal invariance on top of the usual group connected to the identity. However, discrete symmetries are in general not implied by the continuous part of a group. 

That the $S$ symmetry is related to the absence of superselection sectors was first conjectured by Rehren \cite{R6}. Equivalently, given the analysis in \cite{kawahigashi2001multi}, this conjecture relates $S$ symmetry to the validity of Haag duality for general regions. This conjecture was later proved by  Y. Kawahigashi and R. Longo, and independently by M. M\"uger, but the proof remained unpublished, see \cite{muger2010superselection} for a short account. Our intention is then twofold. We first seek to present a mathematical proof of this connection, and we will do so in section \ref{secIV} for the case of rational 2d CFT models. This proof relies ultimately on the algebraic theory of superselection sectors \cite{Doplicher:1971wk,Doplicher:1973at,Doplicher:1990pn,Longo:1994xe}, to be briefly reviewed below. On the other hand, we also want to clarify this idea from a more physical point of view, in part with a view to possible generalizations to higher dimensions. In short, the physical, not mathematically rigorous, argument, is the following. Haag duality violations correspond to scenarios where the algebras ${\cal  A}(R), {\cal A}(R')$ for complementary regions are not commutants of one another. For a pure global state, Renyi entropies are equal for commutant algebras, rather than for complementary regions. Then HDV  should lead to vacuum Renyi entropies for complementary regions that are not equal to each other. This way we will see how to connect Renyi entropies, which are associated with replica partition functions, with the idea of completeness. Quite insightfully, for the particular case of $d=2$ CFT's and Renyi entropies of index $n=2$, the entropy for two intervals can be written in terms of the torus partition function \cite{furukawa2009mutual,calabrese2009entanglement,Headrick:2010zt}. It then follows that the absence of HDV sectors is reflected on the $S$ invariance of the partition function. It also follows we can relate HDV and completeness with aspects of the partition function alone. These results will be covered in sections \ref{secII} and \ref{secIII}.

Before starting we want to make some remarks. Paralleling the situation with group and subgroups, whenever we have an inclusion of algebras $\mathcal{A}\subset \mathcal{B}$, we have an associated index, the ``Jones index'' \cite{J,KOSAKI1986123,L11,L12}. This index measures the relative size of both algebras. It is well defined even for type II and III algebras, e.g. the properly infinite cases. From the present QFT perspective, this index can be applied to the inclusion \ref{ff}, providing a rigorous measure of the amount of HDV/GS. For 2d chiral nets, this index was studied in \cite{kawahigashi2001multi}. There it was directly related to the category of DHR superselection sectors, equivalently the modular tensor category associated with the chiral algebra. In this case, the index turns to be equal to the total quantum dimension of the category. Below we will recover this result using more standard Renyi entropies, in particular differences for Renyi entropies of complementary regions. These differences will also relate to differences of certain limits of torus partition functions. For similar reasons, the index also appears in the topological entanglement entropy associated with topological field theory in $d=3$ \cite{Kitaev:2005dm,Levin:2006zz}. These theories are described by the same modular tensor category as $d=2$ models.  The index of the inclusion \ref{ff} in higher dimensional applications, associated with different types of topologies for $R$, is in general given by the order of a group, and can also be obtained from a limit of entropies \cite{Casini:2020rgj,Magan:2020ake}.

Finally, the appearance of the index in certain limits of the thermal partition function allows to answer the original question, namely the relation between HDV/GS and bootstrap data, for the type of incompleteness related to standard global symmetries in higher dimensions as well. More concretely, this allows to diagnose if a model is the neutral part of another one under the action of an internal symmetry group by the knowledge of the spectrum asymptotics. 
We will discuss this in Sec. \ref{secV}. These developments leave open the interesting question of how HDV associated with non-trivial and not disconnected topologies (e.g. those HDV arising from Wilson/'t Hooft loops in four dimensions), is encoded in the bootstrap CFT data. We will briefly discuss this, along with further open problems and connections, in the discussion section \ref{SecVI}.

\section{Renyi entropies and modular invariance}
\label{secII}

Consider a $2d$ CFT on the space-time cylinder. For the algebra $\mathcal{A}(R)$ of a double cone associated with a spatial interval $R$ there is Haag duality \cite{buchholz1990haag,BGL}. Consider two double cones associated with two disjoint intervals $R=R_1\cup R_2$ at time $t=0$, and define the additive algebra of the union as the minimal algebra containing both, namely ${\cal A}(R)={\cal A}(R_1)\vee {\cal A}(R_2)$. This is the algebra generated by operators localized in $R_1$ and $R_2$. The causal complement $R'$ of $R$ is another two interval region in the surface $t=0$, $R'=R_3\cup R_4$. There is a violation of Haag duality for the two intervals if  ${\cal A}(R)\subset \hat{{\cal A}}(R)\equiv ({\cal A}(R'))'$ is a strict inclusion of algebras. This occurs when there are charged sectors for the theory (DHR sectors) because a charge-anticharge pair belongs to $\hat{{\cal A}}(R)$ but does not belong to the algebra of the two intervals ${\cal A}(R)$, see \cite{kawahigashi2001multi,Casini:2019kex} for detailed discussions.

Since HDV is a phenomenon associated with algebras, it is natural to look for signals of HDV in quantum information measures. These information measures typically input states and algebras and output numbers.   
General entropic order parameters for HDV have been introduced in \cite{Casini:2020rgj}. These are relative entropies whose definition is based on the existence of the two different algebras ${\cal A}(R)\subset ({\cal A}(R'))'$. See also \cite{longo2018relative,Xu:2018uxc}. In a certain pinching limit in which the appropriate size of the region goes to zero/infinity, the entropic order captures the index of the inclusion. As their name suggests, these entropic order parameters serve to characterize phases in QFT. In particular, in quite a unified fashion, they characterize the existence of symmetry and symmetry breaking, including confinement and other forms of GS, such as higher-form symmetries and the non-invertible symmetries that appear in 2d. The entropic order parameters obey a ``certainty relation'' that links the relative entropies on complementary regions \cite{Magan:2020ake,Casini:2019kex,hollands2020variational,Xu:2018uxc}. The connection of these relative entropies with modular invariance has been established in \cite{Xu:2018uxc}. Generalizations of these order parameters to Renyi-type relative entropies have been described \cite{hollands2020variational}. All these entropic quantities (entropic order parameters) are defined via relative entropy. This makes them directly applicable to QFT. In particular, they are UV finite and independent of the regularization techniques. But for our present purposes, it will be convenient to use other quantities directly related to ordinary Renyi entropies. These will allow us to connect HDV with partition functions and the CFT data through the replica trick.\footnote{At any rate, once we arrive at the precise connection between Renyi entropies, partition functions, and the quantum dimensions of the model, we can go back and relate them to the entropic order parameters, whose dependence on the quantum dimensions  has already been established \cite{Casini:2020rgj,longo2018relative,Xu:2018uxc,Magan:2020ake,Casini:2019kex,hollands2020variational}.}

\subsection{Renyi mutual information}

To start we can focus on the $n$-th Renyi mutual information for two intervals $R_1$ and $R_2$. This is defined by
\be
I_n(R_1,R_2)=S_n(R_1)+S_n(R_2)-S_n(R_1\cup R_2)\,,
\ee
with $S_n$ the $n$-th Renyi entropy. To compute the Renyi entropies a cutoff is introduced. The entropies diverge logarithmically with the cutoff by an accumulation of entanglement at the endpoint of the intervals. We do not need to specify the cutoff further. Cutoff details disappear from $I_n(R_1, R_2)$ in the continuum limit.  Then $I_n$ is a conformal invariant quantity that depends on the two intervals through a cross ratio. It will be convenient to map the cylinder to the plane, and the two intervals to $R_1=[a_1,b_1]$,  $R_2=[a_2,b_2]$ at $t=0$, with ordered $a_1<b_1<a_2<b_2$. For later convenience, we write the mutual Renyi entropy in the form 
\be
I_n(R_1,R_2)=I_n(x)= -\frac{(n+1)\, c}{6 n}\, \log(1-x)+ U_n(x)\,,\label{uno}
\ee
where $c$ is the central charge, and the $x$ is the cross ratio
\be
x=\frac{(b_1-a_1)(b_2-a_2)}{(a_2-a_1)(b_2-b_1)} \in (0,1)\,.
\ee
The limit $x\rightarrow 0$ corresponds to very separated intervals. In this limit, we have $I_n(x)\rightarrow 0$ by clustering of correlation functions. Then in this limit $U_n(x)\rightarrow 0$. 

With a cutoff in place, such as a lattice, the Renyi entropies are equal for commutant algebras in a global pure state. This gives the same Renyi entropies for ${\cal A}(R)$ and $\hat{{\cal A}}(R')$. If there is Haag duality for two intervals $\hat{{\cal A}}(R')={\cal A}(R')$, then one recovers the more usual identity $S_n(R)=S_n(R')$. Taking into account that the single interval entropy is \cite{calabrese2004entanglement}
\be
S_n(R)=\frac{(n+1)\, c}{6 n}\, \log(r/\epsilon)\,,
\ee
where $r$ is the interval length and $\epsilon$ the cutoff distance, and assuming $\hat{{\cal A}}(R')={\cal A}(R')\rightarrow S_n(R)=S_n(R')$, we then get the relation 
\be 
U_n(x)=U_n(1-x)\,,  \quad (\textrm{Haag duality for two intervals}).\label{hd}
\ee
The symmetry of this function $U(x)$ then signals Haag duality for two intervals in the Renyi mutual information $I_n$. This symmetry has been checked in several modular invariant theories \cite{furukawa2009mutual,Casini:2005rm,calabrese2011entanglement,alba2010entanglement,Headrick:2010zt}. But it fails in simple examples like a chiral scalar \cite{Arias:2018tmw} or subnets of free fermions \cite{longo2018relative}. In particular, this property implies $\lim_{x\rightarrow 1} U(x)=U(0)=0$.

The usefulness of Renyi entropies is that they are described by partition functions in a replicated manifold when the index $n$ is an integer. One has to evaluate the partition function on $n$ copies of the Euclidean plane sewn on the region $R$ in cyclic order. A particular simplification occurs for two intervals in a $2d$ CFT when the index is $n=2$. In this case, the replica manifold has genus one. It can then be conformally transformed to a torus. This is described in detail in \cite{furukawa2009mutual,Headrick:2010zt,Cardy:2017qhl}, based on earlier work on orbifolds models \cite{lunin2001correlation}. The conformal transformation of the partition function from the flat replicated space to the torus gives an anomalous contribution that needs to be accounted for. Its contribution follows by evaluating the transformation with the Liouville action. This contribution is universal and depends only on the central charge and the transformation properties of the stress tensor. We therefore do not expect any changes depending on whether the theory is complete or not. We will check this expectation more explicitly in section \ref{compu}, where we also discuss special issues that appear for the specific case of the Virasoro net (without any primaries) for $c>1$.  Using this conformal mapping, the general result for  $I_2(x)$  is 
\be
 I_2(x)=\log \, Z[i l]-\frac{c}{12}\, \log\left(\frac{2^8\,(1-x)}{x^2}\right)\label{mutualr2}\,,
 \ee
 or
\be
U_2(x)=\log \, Z[i l]+\frac{c}{6}\, \log\left(\frac{x\,(1-x)}{2^4}\right)\,.\label{unasola}
\ee
Here $Z[\tau]=\textrm{tr} q^{L_0-c/24} \bar{q}^{\bar{L}_0-c/24}$, $q=e^{i 2\pi \tau}$, is the partition function of the original theory, and 
$Z[i l]$ corresponds to a rectangular torus of modular parameter $\tau=i l$. This modular parameter is related to the original cross ratio $x$ by
\be
x=\left(\frac{\theta_2(i l)}{\theta_3(i l)}\right)^4\,,
\quad
1-x=\left(\frac{\theta_2(i/l)}{\theta_3(i/l)}\right)^4\,,\quad l=\frac{_2F_1(\frac{1}{2},\frac{1}{2},1,1-x)}{_2F_1(\frac{1}{2},\frac{1}{2},1,x)}\,.\label{tres}
\ee
 Note $x\leftrightarrow 1-x$ corresponds to $l\leftrightarrow 1/l$. 
The far away interval limit $x\rightarrow 0$ corresponds to the zero temperature limit $l\rightarrow \infty$, 
\be
x\sim 16\, e^{-\pi\, l}\,,
\ee
and the limit in which the two intervals approach each other $x\rightarrow 1$ corresponds to the high temperature one $l \rightarrow 0$
\be
x\sim 1-16\, e^{-\pi/l}\,.\label{29}
\ee

Using (\ref{unasola}) we arrive at the following formula for the ``crossing asymmetry''
\be
A_2(x)=U_2(x)-U_2(1-x)= \log Z[i \,l]- \log Z[i/l]\,. \label{asym}
\ee
We see that Haag duality for two intervals implies through (\ref{hd}) the $S$ modular invariance of the torus partition function $Z[\tau]=Z[-1/\tau]$.\footnote{The generalization of this relation for non rectangular torus can be obtained similarly from Haag duality for intervals boosted with respect to each other.} The converse is also true as will become clear in what follows. 

Summarizing, the failure of crossing symmetry (\ref{hd}) for the Renyi entropies is because Renyi entropies are equal for commutant algebras in a pure state, rather than for complementary regions. The algebras of complementary regions are different from the commutant algebras precisely when Haag duality is violated. The relation $S_n(R)=S_n(R')$, which is connected to modular invariance, is not implied by unitarity. Modular invariance is then a form of completeness of the theory, since Haag duality for generic regions implies we cannot non-trivially extend the given theory without violating causality.

\subsection{Renyi entropies and replica twists}

Another perspective on the connection between modular invariance and completeness comes from the analysis of Renyi entropies themselves and the associated replica twist operators. Renyi entropies are computed through partition functions on replicated manifolds glued along the regions of interest. These replica partition functions can be written in terms of expectation values of replica twists.  More concretely suppose we have a general theory ${\cal T}$. We can form a replicated model ${\cal T}^{\otimes n}$ by taking $n$ independent copies. The replicated model has a global unbroken symmetry given by any permutation group between the copies. In particular, we can take cyclic permutations giving a group $G=\mathbb{Z}_n$ corresponding to the usual Renyi replica symmetries. Choose a generating element $g\in G$. Associated to this symmetry there are unitary operators $\tau_R$ on ${\cal T}^{\otimes n}$ that implement the cyclic symmetry $G$ on $R$ and do nothing in $\bar{R}\subset R'$. $\bar{R}$ is essentially the complementary region $R'$ except for a small buffer zone of width $\epsilon$ at the boundary.\footnote{See Fig 1 in \cite{Benedetti:2022zbb} and associated text for a more detailed description.} This $\epsilon$ serves as a cutoff, equivalently a regularization that allows the twist to be well defined. These twists can be constructed in general by the standard construction \cite{buchholz1986noether}, such that they respect the group operation and $\langle \tau_R \rangle=\langle \tau_R^{-1}\rangle$.
The Renyi entropies are given by
\be
S_{n}(R)=\frac{1}{1-n} \, \log \textrm{tr} \rho_R^n = \frac{1}{1-n} \, \log \langle \tau_R \rangle.
\ee
The operator $\tau_{\bar{R}}=g\, \tau_R^{-1}$ is a replica twist for the complementary region $\bar{R}$. As the vacuum is invariant under $g$, we also have $\langle\tau_{\bar{R}}\rangle=\langle \tau_R\rangle$. This naively suggests the general equality of Renyi entropies for complementary regions. 

However, let us consider more carefully what is going on in the case we have HDV sectors for $R$ in the original theory $\mathcal{T}$. The replicated theory ${\cal T}^{\otimes n}$ will have non local sectors that are the $n$ times tensor product of the original ones. 
 The replica $\mathbb{Z}_n$ 
 symmetry interchanges non local operators of different copies, and then it does not keep the HDV sectors of ${\cal T}^{\otimes n}$ invariant. Equivalently, the HDV sectors of the replicated theory are ``charged'' under the replica symmetry. We are then in the scenario studied in \cite{Benedetti:2022zbb}. 

When $R$ has HDV sectors that transform under the symmetry we have different possible twists for the region $R$ in the replicated model \cite{Benedetti:2022zbb}. These twists differ in their macroscopic properties rather than in regularization. We have complete twists $\tau^c_R$, that implement the symmetry on $\hat{{\cal A}}(R)$, namely, they transform both the local and the non local operators in $R$. We also have additive twists $\tau^a_R$, that only implement the symmetry on the local operators ${\cal A}(R)$ and commute with $\hat{{\cal A}}(\bar{R})$. As their name suggests, the additive twists can be constructed additively in the appropriate region (region $R$ plus buffer zone $\epsilon$), while complete twists do contain HDV operators in themselves. It is important to note that both classes of twists can be constructed. Indeed, starting from additive/complete twists in $R$, the complete/additive twists for the complementary region $\bar{R}$ can be obtained as
\be
\tau^a_{\bar{R}}=g\, (\tau_R^c)^{-1}\,,\quad  \tau^c_{\bar{R}}=g \, (\tau_R^a)^{-1}\,.\label{212}
\ee
They correspond to the group element $g$ in the complementary region. Additive and complete are interchanged under this operation. As the vacuum is invariant under replica symmetry we have $\langle \tau_R^a\rangle=\langle\tau_{\bar{R}}^c\rangle$. This implies that the Renyi entropies of the additive algebras coincide with the Renyi entropies for the complementary region computed with the complete twist. This corresponds to the Renyi entropy for the maximal algebra $({\cal A}(R))'=\hat{{\cal A}}(R')$, but not to the Renyi entropy of ${\cal A}(R')$. This explains why Renyi entropies for complementary regions (and the additive algebra) do not coincide in these anomalous cases. 

In $2d$ CFT this observation may lead to some additional confusion. Take again a region $R$ formed by two disjoint intervals. Twist operators, considered in orbifold theory (the fix point algebra under the action of the symmetry group) are usually thought of as the product of new local field operators at the end-points of the intervals. The reason for this is that the original local charged fields that do not commute with the twist are absent in the orbifold. We can then choose to include the twist in the set of local fields. This process of projecting to the neutral part and adding the twisted sectors is called the orbifold construction in the literature of 2d CFT's. This process takes us from one local theory to another one with a different local field operator content. In this case, in the orbifold theory ${\cal T}^{\otimes n}/\mathbb{Z}_n$ we could write the additive twist $\tau_R^a$ in the following manner
\be
\tau_R^a\,\rightarrow\,\sigma_n(a_1)\sigma_n^\dagger(b_1)\sigma_n(a_2)\sigma_n^\dagger(b_2)\,.\label{deco1}
\ee 
These $\sigma_n(x)$ are thought as local field operators in the orbifold theory, and the additive twist has then a nice and simple expression as a product of four $\sigma_n(x)$ located at the end of the intervals. This is the standard construction of replica partition functions.

However, in the present situation in which the original theory $\mathcal{T}$ has HDV sectors, this rewriting in terms of local operators is not possible, even for the additive twist. The ultimate reason (both physical and mathematical) is that replica symmetry in the replicated theory must change the HDV classes. Then we have both types of twists $\tau_R^c$ and $\tau_R^a$, complete and additive respectively. After projecting to the neutral sector they belong to the neutral algebra, where also the global element $g$ is mapped to the identity. Then, by (\ref{212}) the expression (\ref{deco1}) would also represent the complete twist in the complement of the two intervals. But this is not possible since the complete twist cannot be written additively in the two intervals. The reason is that complete twits still contain non local operators in the orbifold model. This is required in order to move the HDV classes, see \cite{Benedetti:2022zbb} for examples and further explanations.

There is a related perspective on the problem that arises by thinking of the twists as products of local fields $\sigma_n$ located at the endpoints of the intervals. In fact, if expression (\ref{deco1}) is correct, one may obtain the equality of complementary Renyi entropies ($x\rightarrow 1-x$) by the Euclidean crossing invariance of this multipoint field correlator \cite{Headrick:2010zt,Cardy:2017qhl}. As we have seen, this is not correct for models that violate Haag duality for two intervals. 

Let us remark that this reasoning is valid more generally, not only for the case of orbifolding replica symmetry. For any theory with a global symmetry in $2d$, the usual orbifold construction, that follows by taking only the uncharged operators and promoting twists generating the symmetry to local fields, is obstructed when this global symmetry changes HDV classes of the theory. This can be thought as a mixed anomaly between the global symmetry and the generalized symmetry given by the existence of non local classes. However, non local classes will not exist if the initial model is modular invariant, and this is why this obstruction has not been noticed (to our knowledge) in the literature in these terms.

Other heuristics ways of seeing the incompatibility of orbifolding a symmetry that changes HDV classes go by using the original/standard definition of orbifolds $\mathcal{T}/G$ in terms of partition functions \cite{Dixon:1985jw,Dixon:1986jc}, or as gaugings of generalized symmetries \cite{Frohlich:2009gb,Gaiotto:2014kfa,Bhardwaj:2017xup}. In the first approach one projects into invariant states under the symmetry action. But also notices that one can quantize the theory in different sectors, by imposing different boundary conditions on the fields. These new boundary conditions just reflect the orbifold nature of the theory, namely, we should equate field configurations that differ by the action of a certain group element. The proposal in \cite{Dixon:1985jw,Dixon:1986jc} is that this construction produces a modular invariant theory. But this is indeed so when the input theory $\mathcal{T}$ is modular invariant itself. This is implicitly assumed as far as we can tell. But in the light of the present article, it is natural to ask for the result of ``orbifolding'' a not modular invariant theory. As we saw this theory contains HDV sectors. We explored this orbifolding in simple examples concerning subfactors of minimal models and where the symmetry interchanges different HDV sectors, and got unphysical results, as expected. According to our discussion, if a suitable notion of orbifold can be established for non-modular invariant models, it should be for symmetries that do not change sectors. Notice that, as the non local operators (charge-anticharge pairs) are ultimately operators in the local algebra of an interval, the action of the symmetry on the non local sectors is completely determined by the action of the symmetry on the local operators. This seems to suggest that the origin of the problem can be rephrased in terms of 't Hooft anomalies or the mixing of the symmetry with HDV sectors. This parallels the analysis of the ABJ anomaly \cite{Benedetti:2023owa}. It would be interesting to explore this perspective further.

\section{Partition functions and the index}
\label{secIII}

The Renyi crossing asymmetry $A_2 (x)$, eq. (\ref{asym}), measures both the failure of Haag duality for two intervals and the failure of modular invariance in the torus partition function. In general, this asymmetry is a complicated function of the modular parameter. However, we now see that the limit of two touching intervals $x\rightarrow 1$, or equivalently the high temperature limit $l\rightarrow 0$, has a particularly simple expression in terms of the Jones index \cite{J,KOSAKI1986123,L11,L12} associated with the category of representations of the 2d CFT. This occurs when this limit exists or equivalently when the index is finite. This index has a simple expression in terms of the quantum dimensions of the model. It measures the size of the HDV in a universal algebraic manner. To introduce this idea we first consider a simple case where there are two models (in any dimensions) related by a symmetry group.\footnote{For the entropic order parameters \cite{Casini:2020rgj} the relation between the limit of two touching intervals, or more generally the limit in which the size of the regularizing region goes to zero, and the Jones index was established by different means and in different scenarios in \cite{Casini:2020rgj,Casini:2019kex,longo2018relative,Xu:2018uxc,hollands2020variational}.}    

\subsection{The case of a symmetry group}

Suppose we have a theory ${\cal C}$ and a finite global internal symmetry group $G$ acting on ${\cal C}$. The neutral operators invariant under $G$ form another theory ${\cal T}\subset {\cal C}$.\footnote{The logic of the notation, that will extend along the article, is that ${\cal C}$ corresponds to a ``complete'' extension of the theory ${\cal T}$ we are interested in.} The relation between ${\cal C}$ and ${\cal T}$ can be formalized through a conditional expectation $\varepsilon:{\cal C}\rightarrow {\cal T}$ as
\be
\varepsilon (x)=\frac{1}{|G|}\sum_{g\in G} U(g)\, x\, U(g)^\dagger\,, \label{ce}
\ee
where $U(g)$ is the group representation of $G$ on the Hilbert space ${\cal H}_{\cal C}$ of ${\cal C}$. The latter contains charged operators in different sectors, corresponding to the irreducible representations $r$ of $G$. The sector corresponding to the identity coincides with ${\cal T}$. 

Consider the thermal partition functions $Z_{\cal C}(\beta)$ and $Z_{\cal T}(\beta)$, of theories ${\cal C}$ and $\mathcal{T}$ in some compact manifold $M$. These correspond to the path integral over $M\times S^1$. We are interested in the ratio $Z_{\cal T}(\beta)/Z_{\cal C}(\beta)$. The partition function $Z_{\cal T}(\beta)$ only sums
over ${\cal H}_{\cal T}$, the Hilbert space of neutral states, and we are assuming the Hamiltonian $H$ is a neutral operator. This can be obtained from the full Hilbert space ${\cal H}_{\cal C}$ by projecting over the neutral states. To this end we can use the following projection operator
\be
P=\frac{1}{|G|} \sum_{g\in G} U(g)\,. \label{proj}
\ee
 Then we have
\be\label{quo}
\frac{Z_{\cal T}(\beta)}{Z_{\cal C}(\beta)}=\langle P\rangle_\beta= \frac{1}{|G|}\sum_{g\in G} \langle U(g)\rangle_\beta\,,\quad \langle \cdots\rangle_\beta\equiv \frac{\textrm{tr}_{\cal H_{\cal C}} e^{-\beta H} (\cdots)}{Z_{\cal C}(\beta)}\,.
\ee
In the path integral formulation $Z_{\cal T}(\beta)$ corresponds to the original path integral for ${\cal C}$ with insertions of the group elements at time $t=0$, and averaged over the group.  This ratio will in general be a complicated function of the temperature and the size of $M$.

A simplification occurs in the high temperature limit. In this limit, there are large fluctuations of charges that make 
\be
\lim_{\beta\rightarrow 0} \langle U(g)\rangle_\beta\rightarrow \delta_{g,1}\,.
\label{fc} \ee
This was proved in \cite{magan2021proof} using algebraic methods. In physical terms, at high temperature, we have a large number of independent charged fluctuations of different irreducible representations $r$ of the group at different points in $M$. Then the typical global representation will be the tensor product of a large number of different random representations. 
The product of representations can be thought as a stochastic process that converges to multiples of the regular representation (see for example \cite{Casini:2019kex}). The character of the regular representation is proportional to $\delta_{g,1}$, and this gives zero expectation values for group elements different from the identity in these representations.
This gives
\be
\lim_{\beta\rightarrow 0}\frac{Z_{\cal T}(\beta)}{Z_{\cal C}(\beta)}=\lim_{\beta\rightarrow 0} \langle P\rangle_\beta= \frac{1}{|G|}\,.\label{lab1}
\ee
Hence, this ratio of partition functions in the high temperature limit is purely kinematical. It allows us to evaluate the size of the group.

We can generalize $Z_{\cal T}(\beta)$ by evaluating the partition function $Z^{(r)}(\beta)$ restricted to global states in any given irreducible representation $r$ (so $Z_{\cal T}(\beta)=Z^{(1)}(\beta)$). In this case, we have to insert the projector on the representation $r$
\be
P_r=\frac{d_r}{|G|} \sum_{g\in G} \chi_r^*(g) \, U(g) \label{proj1}
\ee
in the partition function. Using (\ref{fc}) we get 
\be
\lim_{\beta\rightarrow 0}\frac{Z^{(r)}(\beta)}{Z_{\cal C}(\beta)}=\lim_{\beta\rightarrow 0} \langle P_r\rangle_\beta= \frac{d_r^2}{|G|}\,.\label{lab}
\ee
Hence, the  probability of the charge $r$ at high temperature is given by  
\be
p_r=\frac{d_r^2}{|G|}\,. \label{fff}
\ee
This probability distribution was determined in \cite{Casini:2019kex}, both as the probability (or ratio) of different irreducible representations in the vacuum density matrix or in the thermal density matrix at infinite temperature. Ref. \cite{Pal:2020wwd} independently arrived at it in the context of $d=2$ CFT's, where it appears as controlling the contribution to the density of states of a specific irreducible representation. In this vein it was later conjectured in \cite{harlow2022universal}. Ref. \cite{magan2021proof} established the equivalence between all these perspectives and presented a direct general proof. See also \cite{Cao:2021euf} for a derivation in weakly coupled theories.

Notice this $p_r$ is exactly the fraction of basis states in the representation $r$ within the regular representation, as expected from previous arguments. These results translate quite directly to other situations. For example, we can obtain the same results by taking the limit of large volume at any fixed non zero temperature, even for a massive theory. The reason is again the existence of a large number of independently charged fluctuations with non zero probabilities. 

Another scenario is to consider vacuum expectation values of twists operators. A twist is a unitary operator $\tau(g)$ that implements the symmetry operation in some region of the space $R$ and does nothing in a spatially separated region $\bar{R}$.\footnote{Strictly speaking, in the causal region determined by $R$.}  For concreteness we take the region $R$ to be a ball and $\bar{R}$ the complement of a ball separated by a distance $\epsilon$ form $R$ on the  surface $t=0$. Twists can be constructed with standard methods such that they are representations of the group $G$, and transform covariantly under the global operations $U(h) \,\tau(g)\, U(h)^\dagger=\tau(h\,g\,h^{-1})$ \cite{Doplicher:1984zz,buchholz1986noether}. If we take a decreasingly small separating size  $\epsilon$, the twists become very sharp. They will necessarily display large fluctuations except for the identity. In the limit one then obtains 
\be
\lim_{\epsilon\rightarrow 0} \langle \tau(g)\rangle\rightarrow \delta_{g,1}\,. \label{ty}
\ee
 This can also be understood as a consequence of the existence of a large number of almost independent, charge-anticharge pair fluctuations, with non trivial probability in the vacuum, where charge and anticharge operators are localized in $R$ and $\bar{R}$, separated by a distance of order $\mathcal{O}(\epsilon)$ across the boundary of $R$.      
 A formula analog to (\ref{lab}) (and to  (\ref{lab1}) for $r=1$) in this case involves the Euclidean plane partitions function $Z^{(r)}(R,\epsilon)$ with insertions of localized charge projectors $P_r(R,\epsilon)$ in $R$ constructed in terms of the local twists with the formula analog to (\ref{proj1}). The limit of large temperatures is replaced by the limit $\epsilon\rightarrow 0$.  More explicitly
\be
\lim_{\epsilon \rightarrow 0} \frac{Z^{(r)}(R,\epsilon)}{Z_0}=\lim_{\epsilon \rightarrow 0} \langle P_r(R,\epsilon) \rangle=\frac{d_r^2}{|G|}\,, \label{lali}
\ee 
where $Z_0$ is the plane partition function without insertions, and the expectation value is in vacuum. This is the way these probabilities appear in the entanglement structure of the vacuum, first obtained in \cite{Casini:2019kex}, see also \cite{Kusuki:2023bsp} for a different approach to the same formula in 2d CFT. We note these formulas make the connection with symmetry resolved entropies \cite{Xavier:2018kqb,Murciano:2020vgh} transparent, as explained in \cite{magan2021proof}.

\subsection{More general charged sectors}
\label{drd}

The case treated previously corresponds to ${\cal T}\subset {\cal C}$, where ${\cal T}$ is the fix-point subalgebra of ${\cal C}$ under the action of an internal symmetry group $G$. According to DHR theorem \cite{Doplicher:1971wk,Doplicher:1973at,Doplicher:1990pn,Longo:1994xe}, this is the general case of a subtheory in dimensions $d>2$. In $d=$ we can have more general scenarios. More precisely, the relation between an algebra of charged sectors of a theory ${\cal C}$ and a subtheory ${\cal T}$ is not necessarily associated with the average of a group. Incomplete theories in $d=2$ give rise in general to braided categories, as opposed to symmetric categories.  See \cite{haag2012local} for a discussion and references. 

However, much of the relevant structure is common to both cases. In particular, we have again irreducible charged sectors $r$, $r\in {\cal R}$, where $r=1$ is the identity, and we are assuming for simplicity to form a finite set. Operators in each of the charged sectors are transformed within themselves under the operators in $\mathcal{T}$. Hence, a general element of ${\cal C}$ can be written
\be
c=\sum_{r,i}  t_{r,i}\, \psi_r^i\,, \label{ssa}
\ee
where $t_{r,i}\in {\cal T}$ and $\psi_{r}^i$ can be chosen as fixed elements of ${\cal C}$. The index $i\in 1,\cdots,N_r$, marks a possible multiplicity of the sectors. Accordingly, the Hilbert space decomposes into sectors ${\cal H}=\oplus\, N_r\,  {\cal H}_{r}$, and we have projectors $P_r$ associated to each sector. 

Because ${\cal C}$ is a closed algebra we have and OPE-like decomposition \cite{Longo:1994xe} 
\be
\psi_{r_1}^{i_1}\, \psi_{r_2}^{i_2}= \sum t_{r_1,i_1,r_2,i_2,r_3,i_3}\, \psi_{r_3}^{i_3}\;,
\ee
for some elements $t_{r_1,i_1,r_2,i_2,r_3,i_3}\in {\cal T}$. This suggest the idea of a fusion of the different sectors. To express this idea with more precision, and to introduce the dimensions of the sectors, it is convenient to choose  charged operators in (\ref{ssa}) that are isometries
\be
\psi_r^{i\,\dagger} \,\psi_r^j=\delta_{ij}\,.
\ee
Then, changing the ordering of the neutral and charged operator in an expansion like (\ref{ssa}) gives
\be
 \psi_r^i\, t= \rho_r(t)\, \psi_r^i\,,
\ee
where $t\rightarrow \rho(t)$ is an endomorphism of the algebra ${\cal T}$. This way, charged sectors are represented as special endomorphisms of the neutral algebra that are induced by the inclusion ${\cal T}\subset {\cal C}$.
The importance of the endomorphisms is that they can be composed (fusion) and can be decomposed into direct sums of irreducibles. This is analogous to the theory of representations of a group. It is to be noted however, that the fusion of endomorphism $\rho_r$, $r\in {\cal R}$, in general will generate a bigger set of irreducible endomorphisms $\tilde{{\cal R}}$ of $\cal T$. In this larger class the fusion can be written in a standard form 
\be
 \rho_{r_1}\,\circ \,\rho_{r_2}= \bigoplus_{r_3\in \tilde{{\cal R}}} N^{r_3}_{r_1 r_2} \,\, \rho_{r_3}\,,
\ee 
where $N^{r_3}_{r_1 r_2}\in \mathbb{N}$, and  $N_{1 r_1}^{r_2}=\delta_{r_1,r_2}$. For each sector $r$ there is a unique conjugate sector $\bar{r}$ which is the only one that gives a fusion containing the identity, namely $N_{r s}^1=\delta_{r \bar{s}}$. The fusion is associative and commutative, and then the matrices $N^{(r)}$ defined by
\be
N^{(r)}_{st}=N_{rt}^s
\ee 
commute with each other. They are then simultaneously diagonalizable. Since they have positive entries, by the Perron-Frobeniuos theorem there is a unique common eigenvector with positive entries and positive eigenvalues. These eigenvalues are also the maximal ones for each matrix. The maximal eigenvalue of $N^{(r)}$ is called the dimension $d_r$ of the sector (or the quantum dimension alike).  It also follows that the common eigenvector with positive entries for the matrices $N^{(r)}$ is proportional to $(d_1,d_2,\cdots,d_n)$. We have $d_r\ge 1$, and $d_1=1$. A non irreducible representation with a stable fractional irreducible content under fusion with any other one, analogous to the regular representation for a group, satisfies 
\be
\rho_r\times \sum_s q_s\,\rho_s= \sum_{st} N^{(r)}_{st} q_s \, \rho_t\propto \sum_t  q_t\, \rho_t\,,\quad  r=1,\cdots,n \,. 
\ee
Then, it must be the case that $q_r\sim d_r$, is proportional to the eigenvector of the fusion matrices with positive entries. In fact, the fusion of many sectors always converges to this distribution $q_r\sim d_r$. This gives the fraction of representations in the stationary distribution.

In $d=2$ the dimensions are in general not integer numbers. A physical understanding of these non integer dimensions is that they determine the asymptotic dimension of the Hilbert space associated with $m$ excitations of a given sector $r$, when $m$ goes to infinity, see \cite{PreskillNotes}. Indeed, due again to the Perron-Frobeniuos theorem, the dimension of the intertwiner space associated with $m$ sectors of type $r$ scales proportionally to $d_r^m$ for large $m$. 
Another interpretation of the quantum dimensions that does not involve the fusion matrices was developed by R. Longo \cite{L11,L12}. There the dimensions arise from the Jones index associated with the inclusion of algebras determined by the endomorphism of $\mathcal{T}$ induced by the given sector $r$. 

Now, consider again the theory in a thermal state at high temperature in a compact space $M$ (a circle in $2d$). The full Hilbert space ${\cal H}_{\mathcal{C}}$ is decomposed as a direct sum of spaces of different charges ${\cal H}_{r}$, and we have the corresponding projectors to these spaces $P_r$. In particular, we again have ${\cal H}_{1}={\cal H}_{\mathcal{T}}$, the Hilbert space of ``neutral'' states. In analogy with the case of a group symmetry described above, the high temperature limit saturates the probabilities of the different sectors to be proportional to their dimensions and multiplicity
\be\label{prov}
p_r=\lim_{\beta\rightarrow 0} \langle P_r\rangle_\beta= \frac{N_r\, d_r}{\sum_{r\in {\cal R}} N_r\,d_r}\,.
\ee
We can heuristically think this limit as the one stable under the addition of new charged excitations, as for the case of groups. 
The sum 
\be
\lambda=[\mathcal{C}:\mathcal{T}]=\sum_{r\in {\cal R}} N_r\, d_r\:,\label{qqq}
\ee
is called the Jones index of the inclusion $\mathcal{T}\subset \mathcal{C}$. It can be defined in purely algebraic terms \cite{J,KOSAKI1986123,L11,L12}. As mentioned before, this is also called the total quantum dimension in topological models in $d=3$ \cite{PreskillNotes}. For $\mathcal{T}$ the neutral part of $\mathcal{C}$ under the action of a finite group $G$ we have $N_r=d_r$ and then $\lambda=\sum_r d_r^2=|G|$, so that the index coincides with the order of the group. 

To prove the result (\ref{prov}) for $\beta\rightarrow 0$ in a rigorous way more technical tools are needed. This proof, using a thermofield double description, and following the analogous one for group symmetries in \cite{magan2021proof}, is presented in Appendix \ref{app}. For CFT's a different proof can be found in \cite{Lin:2022dhv}. The present algebraic approach only uses relativistic symmetry and it is valid in generic $d=2$ QFT's.

The dual algebra to the charged sectors, analogous to the group operations and the projectors into the different representations in the orbifold case, are usually called topological defect lines or Verlinde lines in the physical literature. For diagonal modular invariants, they were defined in the Verlinde's original contribution \cite{VERLINDE1988360}, but there are more general topological defect lines, see e.g. \cite{Frohlich:2006ch}. At any rate, in the present algebraic formulation, von Neumann's double commutant theorem implies that this dual category of non local operators comes in parallel with that of charged sectors.  Moreover, this dual algebra comes with precisely the same size, the global index described before \cite{Longo:1994xe,Review}.\footnote{For the QFT vacuum and compact regions, this algebraic structure can be nicely described into quantum complementarity diagrams that encode the HDV \cite{Casini:2020rgj,Review,Magan:2020ake}. These complementarity diagrams are part of the Jones ladder \cite{J} associated to the inclusion of algebras ${\cal A}(R)\subseteq ({\cal A}(R'))'$.} Indeed, in the general finite index case we also have a conditional expectation $\varepsilon:\mathcal{C}\rightarrow \mathcal{T}$ that eliminates the charged operators, see \cite{Longo:1994xe} for the algebraic analysis of these more general nets of subfactors. And there is also a dual conditional expectation $\varepsilon'$ that maps the topological defect projectors into the identity\footnote{The projection $e:{\cal H}_{\mathcal{C}}\rightarrow {\cal H}_{\mathcal{T}}$ is called the Jones projection in the mathematical literature, and we have $\epsilon'(e)=\lambda^{-1} \mathbb{1}$. In this paper, we have called $P$ to this projection.}
\be\label{prg}
\varepsilon'(P_r)=\frac{N_r\,d_r}{\lambda} \,\mathbb{1}\;.
\ee
Hence the thermal state $\omega_\beta$ approaches at high temperatures a state invariant under the dual conditional expectation $\omega_\beta\rightarrow \omega_\beta\circ \epsilon'$. 

\subsection{The limit of the asymmetry}
Using these results and assuming a finite index, we now compute a particular limit of the Renyi crossing asymmetry for a generic theory ${\cal T}$
\be
\lim_{x\rightarrow 1} U_2(x)-U_2(1-x)=U_2(1)=\lim_{l\rightarrow 0} \left( \log Z_{\cal T}[i \,l]-\log Z_{\cal T}[i/l]\right)\,. \label{asyma}
\ee
This is zero for a modular invariant or complete model. But now assume the model $\mathcal{T}$ is a submodel of a complete one $\mathcal{C}$, so that $\mathcal{T}\subset \mathcal{C}$ is a strict inclusion. In the complete model the partition function can be pictured as a featureless path integral without insertions, and satisfying modular invariance directly from the geometric symmetry of the calculation. Then we have
\be
U_2(1)=\lim_{l\rightarrow 0} \left(\log \frac{Z_{\cal T}[i \,l]}{Z_{\mathcal{C}}[i \, l]}-\log \frac{Z_{\cal T}[i/l]}{Z_{\mathcal{C}}[i/l]}\right)\,.
\ee
Now, the second term in this expression corresponds to the zero temperature limit. In this limit, the partition function is dominated by the vacuum or Casimir energy, and subleading terms give exponentially small corrections. Since the vacuum and Hamiltonian are the same for the two models this term vanishes in the limit. We get
\be\label{u2m}
U_2(1)=\lim_{l\rightarrow 0} \log \frac{Z_{\cal T}[i \,l]}{Z_{\mathcal{C}}[i \, l]}= -\log \lambda\,.
\ee
This follows from \ref{quo} and \ref{prg}. Since for a complete model we have $U(1)=0$, this result clearly expresses the failure of completeness of the model in terms of the index, the partition function and the Renyi entropy.  Notice we could get individual sector probabilities (\ref{prov}), replacing $\lambda^{-1}$ in this formula, just by considering the projected partition functions into different representations $r$. Following \cite{magan2021proof} (see also appendix \ref{app}), this way one derives symmetry resolved entropies \cite{Xavier:2018kqb,Murciano:2020vgh} and the universal charged density of states for these more generic (typically non-invertible) scenarios. Indeed all these quantities follow from eq. (\ref{prg}) as well. In particular, the density matrix has the block decomposition $\oplus_r p_r\rho_r$, with $p_r=N_r\,d_r/\lambda$, first derived in \cite{Casini:2019kex} for groups, here seen to be valid in generic categorical scenarios. Symmetry resolved entropies for invertible symmetries were also derived using boundary conformal field theory methods in \cite{Kusuki:2023bsp,Saura-Bastida:2024yye}, obtaining the same result.\footnote{In \cite{Kusuki:2023bsp,Saura-Bastida:2024yye} an obstruction to define symmetry resolved entropies was found for theories in which the symmetry cannot be orbifolded. This obstruction is a problem with the computational technique, that requires certain classes of boundary conditions. It is related to the problems mentioned in section \ref{secII}. While the expression of the twists in terms of local fields might not exist (obstructing the existence of the orbifold), the twists and projectors still exist as non local operators in the original model. Therefore formula \ref{u2m}, and associated symmetry resolved entropies, are still well defined and correct even when considered for the computation of entanglement entropy in such anomalous scenarios.} Considerations for the more general categorical scenarios in $d=2$ CFT's appear shortly after publication of the present paper in \cite{Choi:2024wfm,Das:2024qdx,Heymann:2024vvf}, again using boundary conformal field theory.

Still, expression (\ref{u2m}) is not completely satisfactory because it involves the index of the incomplete model $\mathcal{T}$ with respect to a given completion $\mathcal{C}$. Although the index itself does not depend on the particular completion, it would be more appropriate to find an intrinsic formula that only cares about theory $\mathcal{T}$. This can be achieved by using an algebraic index that pertains to the incomplete model itself. A particularly good candidate is the two interval index, or global index $\mu$. This is the Jones index associated with the inclusion of the algebras 
\be
{\cal A}(R)\subset\hat{{\cal A}}(R)\,,
\ee
where $R=R_1\cup R_2$ is the two interval region. 
In \cite{kawahigashi2001multi}, proposition 24, it was shown that for two models $\mathcal{T}\subset \mathcal{C}$ the global indices of the two models satisfy the following relation
\be\label{twointglob}
\mu_{\mathcal{T}}= \lambda^2\,\mu_{\mathcal{C}}\,,
\ee
where we remind $\lambda=[{\cal C}:{\cal T}]$, namely the index of the inclusion $\mathcal{T}\subset \mathcal{C}$. Then, if $\mathcal{C}$ is complete, $\mu_{\mathcal{C}}=1$, it follows that
 $\mu_{\mathcal{T}}=\lambda^2$. We then obtain the desired expression
 \be
U_2(1)= -\frac{1}{2}\log \mu\,.\label{tt}
\ee
This is an intrinsic expression measuring the failure of completeness/modular invariance for a model $\mathcal{T}$.

The index appears as a subleading term in the free energy $\log Z_{\cal T}$ without involving the partition function of the complete model $Z_{\mathcal{C}}$. To see this consider first the case of the complete model where $\log Z_{\mathcal{C}}[i l]=\log Z_{\mathcal{C}}[i/l]$. The low temperature expansion is dominated by the Casimir energy, and has exponentially small corrections because of the discreteness of the spectrum,
\be
\log Z_\mathcal{C}[i/l]\simeq \frac{\pi\,c l}{6} + \textrm{exponentially small}\,,\quad l\gg 1\;.\label{aba}
\ee
Then it does not have a constant term. In consequence, the same happens for the large temperature expansion 
\be
\log Z_\mathcal{C}[i\, l]\simeq \frac{\pi\,c }{6 l} + \textrm{exponentially small}\label{aba1}\,,\quad l\ll 1\;.
\ee
  Then, from (\ref{u2m}) we get
\be
\log Z_\mathcal{T}[i \,l]=  \frac{\pi\,c }{6 l} -\frac{1}{2}\log \mu+\cdots\,, \quad l\ll 1\,. 
\ee
This formula was obtained (for chiral theories) in \cite{kawahigashi2005noncommutative}. Note the first term gives the Cardy formula \cite{Cardy:1986ie} for the thermal entropy, that  does not depend on exact modular invariance, at least for finite index (see section \ref{infinite} below for the discussion of an infinite index).

In terms of the mutual information (Renyi index $n=1$) we have for the corresponding $U_1(x)$ the same limit in terms of the index \cite{Casini:2019kex,Casini:2020rgj,longo2018relative,hollands2020variational}
\be
U_1(1)=-\frac{1}{2}\log \mu\,.
\ee
The same limit holds for all Renyi entropies, as we discuss in section \ref{compu} below.

It is interesting to note that the mutual information crossing asymmetry, $A_1(x)=U_1(x)-U_1(1-x)$, is monotonically decreasing, at least for models that are submodels of a complete one. If there is a conditional expectation from $\epsilon:\mathcal{C}\rightarrow \mathcal{T}$, then the conditional expectation property of mutual information implies that $I_1^{\mathcal{C}}(x)-I_1^{\mathcal{T}}(x)$ can be written as a relative entropy \cite{ohya2004quantum,petz2007quantum,Casini:2019kex,longo2018relative}. It is precisely a special case of an entropic order parameter. The monotonicity of this relative entropy gives the monotonicity of the mutual information difference between the two models  
\be
(I_1^{\mathcal{C}}(x)-I_1^{\mathcal{T}}(x))'=(U^{\mathcal{C}}_1(x)-U^{\mathcal{T}}_1(x))'\ge 0\,.
\ee
 Then, as $\mathcal{C}$ is complete, $U_1^{\mathcal{C}}(x)-U_1^{\mathcal{C}}(1-x)=0$, leading to
\be
(A^{\mathcal{T}}_1(x))'=(U_1^{\mathcal{T}}(x)-U_1^{\mathcal{T}}(1-x))'\le 0\,. 
\ee
This implies the extremal values of the asymmetry, reached at $x=0,1$, are proportional to the index. Several examples suggest the same property could be valid for the asymmetry $A_2(x)$, and hence also for the asymmetry in the partition function. However, we could not prove this is generally valid.

\subsection{DHR sectors, global index, and statistical modular matrices} 

For a given theory $\mathcal{T}$ there may be localizable superselection sectors. That is, states that differ from the vacuum in a bounded region and that cannot be produced out of the vacuum with the local operators.  These are called DHR sectors \cite{Doplicher:1971wk,Doplicher:1973at}. This is a characteristic of the theory alone. These sectors can be analyzed already in the vacuum sector of the theory. In this particular sector, they are seen as inducing endomorphisms of the local algebras. The main observation is that these endomorphisms form a category, with direct sums, subobjects and conjugates. One of the most important observations is that in $d>2$, this category has a symmetric or permutation exchange symmetry. This key result serves as the starting point for the DHR reconstruction theorem \cite{Doplicher:1990pn,Longo:1994xe}, stating that these sectors are given by the representations of a symmetry group. On the other hand, for $d=2$ the category of DHR endomorphisms does not admit generically a permutation symmetry. Instead, it is a braided category \cite{FRS,cmp/1104179464}, with a representation of the braiding group.

As introduced in the previous section, in both scenarios there is a natural dimension $d_r$ for any given irreducible sector $r$. This dimension was originally defined in \cite{Doplicher:1971wk,Doplicher:1973at}, see also \cite{Longo:1996hkk}. The relation between the algebraic approach in terms of an inclusion of algebras with associated Jones index, as used above, and the DHR endomorphisms was laid down in \cite{L11}. In particular the index of the inclusion $\rho_r (\mathcal{T}) \subset \mathcal{T}$, where $\rho_r$ is an endomorphism with label $r$ is equal to $d_r^2$. For the specific case of the two interval inclusion, Ref. \cite{Longo:1994xe,kawahigashi2001multi} showed the global index to be equal to the dimension of the canonical endomorphism, namely
\be
\mu=\sum_{r\in DHR} d_r^2\,,
\ee
where $r$ runs over all possible DHR sectors. This is natural because non local operators in two intervals are charge anticharge operators in each interval, and the sum must run over all possible sectors.

Notice that in $d=2$ this leads to surprises. If we have a complete model, then the two interval inclusion is trivial and therefore $\mu=1$. Equivalently there is no non trivial DHR sector. But if we have a model $\mathcal{T}$ with index $\lambda$ with respect to a complete model $\mathcal{C}$, then $\mu_{\mathcal{T}}=\lambda^2$, as mentioned above. We can think heuristically that there are more DHR sectors than the ones the complete model $\mathcal{C}$ has with respect to $\mathcal{T}$. The reason is that DHR sectors need not be local to each other in $d=2$. For example, if the relation to a complete model is a group $G$, we have $\lambda=|G|$, and $\mu_{\mathcal{T}}=\sum_{r\in DHR} d_r^2=|G|^2$. The reason is that both charges and twists of the symmetry group give place to $DHR$ sectors. We see that DHR sectors and/or the two interval inclusion contains information about all possible local completions. On the other hand, in more dimensions the index $\mu_{\mathcal{T}}$ of a two double cone region is $|G|$ \cite{L11,Casini:2019kex} since there is no possibility of twist sectors that can be localized in the two balls. This is related with the fact that completions of two-ball HDV are unique in higher dimensions.\footnote{Still, in higher dimensions we can again see differences for other types of regions. For example, pure gauge theories in $d=4$ have an index of the inclusion of algebras associated with a ring which is the square of any possible Haag-Dirac net we might define \cite{Casini:2020rgj}. Here Haag-Dirac nets play the role of a local completion in $d=2$. Still, these ``completions'' of higher form symmetries in higher dimensions are not full-fledged completions since they are achieved at the expense of additivity of the net of algebras.}  

After this digression let's go back to $d=2$. In these QFT's, the braided category of DHR sectors has further important properties \cite{R4}. In particular, there is always associated a unitary representation of the group $SL(2,\mathbb{Z})$, with ``statistical'' matrices $\tilde{S}$, $\tilde{T}$, such that\footnote{We use tildes for the statistical matrices to avoid confusion with the modular transformations matrices associated to chiral characters. These later are the statistical matrices of the special case of the  DHR sectors associated to the chiral subtheory.}
 \be
\tilde{S} \tilde{S}^\dagger=\tilde{T} \tilde{T}^\dagger =1\,,\quad \tilde{T}\tilde{S}\tilde{T}\tilde{S}\tilde{T}=\tilde{S}\,,\quad \tilde{S}^2=C\,,\quad \tilde{T}C=C\tilde{T}=\tilde{T}\,.
\ee
The matrix $C=\delta_{r\bar{s}}$ is the conjugation matrix. Further, there are Verlinde-like formulas \cite{VERLINDE1988360} for the fusion matrices
\be
N^{r}_{st}=\sum_u \frac{\tilde{S}_{su}\tilde{S}_{tu}\tilde{S}^*_{ru}}{\tilde{S}_{0u}}\,,
\ee
and
\be
\tilde{S}_{0r}=\tilde{S}_{r0}=\frac{d_r}{\sqrt{\sum d_r^2}}\,,\quad \tilde{S}_{rs}=\tilde{S}_{sr}=\tilde{S}_{r\bar{s}}^*=\tilde{S}_{\bar{r}\bar{s}}\,, \quad d_r=\frac{\tilde{S}_{r0}}{\tilde{S}_{00}}\,.
\ee
The global index is 
\be
\mu=\sum_r d_r^2=\tilde{S}_{00}^{-2}\,.\label{tiro}
\ee
For groups, there are analog formulas in terms of the group characters, but here $\tilde{S}$ is symmetric in contrast to the case of non-abelian groups. 

\subsection{Modular invariance and index in terms of chiral characters}

Let $\mathcal{A}$ be a local chiral algebra for the light cone. For each superselection sector $r$ of this algebra we have a character defined as 
\be
\chi_r(\tau)=\textrm{Tr}_r \,e^{2 \pi i \tau (L_0-c/24)}\,.
\ee
Assume these transform under modular transformations linearly as
\be 
\chi_r (-1/\tau)=\sum_s S_{rs} \, \chi_s (\tau)\,,\quad \chi_r(\tau+1)=\sum_s T_{rs} \, \chi_s(\tau)\,.
\ee
In all known cases these $S,T$ are the same matrices as the statistical matrices for the DHR sectors of the chiral theory. A chiral net where this holds was called a modular net in \cite{kawahigashi2005noncommutative}, where it was shown that this property is preserved by local extensions of the chiral net. We note the standard relation between characters and dimensions of the modular category
\be
\lim_{\tau\rightarrow 0} \frac{\chi_r(\tau)}{\chi_1(\tau)}=\frac{S_{0r}}{S_{00}}=d_r\,, \quad \lim_{\tau\rightarrow 0} \frac{\chi_1(\tau)}{\chi_1(-1/\tau)}=S_{00}\,. \label{tit}
\ee
See e.g. \cite{kawahigashi2005noncommutative} for a derivation from the present perspective. 

Now, if we take $\mathcal{A}\times\mathcal{A}$ as the full chiral algebra for a CFT, the CFT corresponds to a local extension $\mathcal{A}\times \mathcal{A}\subset \mathcal{T}$ of the chiral algebra. The partition function of the CFT is  
\be\label{partM}
Z=\sum M_{rs}\, \chi_r(\tau)\, \chi_s(\bar{\tau})\,,
\ee
with the coupling matrix $M_{rs}\in\mathbb{N}$, and $M_{00}=1$, by uniqueness of vacuum. 

Modular invariance of $Z$ then corresponds the two matrix relations 
\be 
S\, M=M\, S\,,\quad T \,M=M\, T\,. 
\ee
The meaning of the $T$ symmetry is well known, and for unitary local CFT's this symmetry is obligatory by the locality of the theory \cite{francesco2012conformal,R4}. Let us clarify this shortly here. Given a primary field $\phi=\phi_i\phi_{\bar{j}}$ with scaling dimensions $(h_i,\bar{h}_j)$ the Euclidean correlation function is
\be
\langle \phi(z)\phi(0)\rangle = \frac{1}{z^{2 h}\, \bar{z}^{2 \bar{h}} }\,.
\ee
Euclidean correlations of local fields have to be real analytic everywhere except the coincidence points, and (for bosonic fields) must be permutation invariant. This is only possible if 
\be
h=\bar{h}+k\,,\quad k \in \mathbb{Z}.\label{condition1}
\ee
This is a condition on the self-locality of the field. It coincides with the requirement of $T$ modular invariance for the parity symmetric case $c=\bar{c}$ (for simplicity we are restricting to this case). One can check that self locality of the fields is enough for locality in three and higher point functions provided the primaries close an algebra of self local fields. As an example, we might worry that $T$ duality might not invalidate a scenario in which we have a twist field and a local field charged under it. These two fields are self local themselves but fail to be local relatively to each other. Here is where closing of the algebra with self local fields enters the game. In fact, the product of a twist and a local field charged under it will produce not self local fields.

To find the expression of the global index of the local CFT extension $\mathcal{T}$ defined by the coupling matrix $M$, we use the limit of the asymmetry (\ref{asyma}), (\ref{u2m}), (\ref{tt}). This is given in terms of the partition function as
\be
\mu_{\mathcal{T}}^{-1/2}= \lim_{l \rightarrow 0} \frac{Z(\tau)}{Z(-1/\tau)} =\lim_{l \rightarrow \infty} \frac{Z(-1/\tau)}{Z(\tau)} 
= \lim_{l\rightarrow \infty}\frac{\sum\limits_{rstu} S_{rs} \,M_{st}\, S^*_{tu} \,\chi_r(\tau)\,\chi_u(\bar{\tau})}{\sum\limits_{rs} M_{rs} \,\chi_r(\tau)\,\chi_s(\bar{\tau}) } \,.\label{ZZ}
\ee
Taking into account that in this last, low temperature, limit, in the denominator dominates the character of the identity, that $M_{00}=1$, and that
\be
\lim_{\tau\rightarrow \infty} \frac{\chi_r(\tau)}{\chi_1(\tau)}=\delta_{r,0}\,,
\ee
we get
\be
\mu_{\mathcal{T}}= \left(\frac{\sum_i d_i^2}{\sum_{ij} d_i\, M_{ij}\, d_j}\right)^2\,,\label{ll}
\ee
where $d_i$ are given by (\ref{tit}). 
This number is $1$ for a modular invariant model. $\mu_{\cal T}$ has to be greater or equal to one for a consistent (closed under fusion) model. As a bona fide Jones index, it also has to fall into one of the allowed values \cite{J}. Note that if the model admits a local extension of the chiral algebra, the matrix $M$, the characters, and the dimensions $d_i$ will be different expressed with respect to this new chiral algebra.  But the global index (and the partition function) of ${\cal T}$ remains the same. 

To better understand formula (\ref{ll}) first note that the index $\lambda=[{\cal T}:{\cal A}\times {\cal A}]$ between the model ${\cal T}$ and the chiral subalgebra is
\be
[{\cal T}:{\cal A}\times {\cal A}]= \sum_{ij} d_i\, M_{ij}\, d_j\,.
\ee
This is an instance of (\ref{qqq}), where the different $r$ are in correspondence with the primary fields $\phi=\phi_i\phi_{\bar{j}}$ of the theory formed by product of chiral fields, having dimension $d_i d_j$, and the multiplicities $N_r$ in (\ref{qqq}) are replaced by $M_{ij}$. The global index of the chiral algebra is the product of the global indices of the two chiral parts, each of which is given by  formula (\ref{tiro}) 
\be
\mu_{{\cal A}\times {\cal A}}=\left(\sum d_i^2\right)^2\,.
\ee
Then it follows that (\ref{ll}) is a particular instance of (\ref{twointglob}), giving 
\be
\mu_{{\cal A}\times {\cal A}}= [{\cal T}:{\cal A}\times {\cal A}]^2\, \mu_{\cal T}\,.
\ee

\subsection{Some examples with infinite index}
\label{infinite}

In the case of an infinite index, the Renyi crossing asymmetry is unbounded as we get to the limits $x\rightarrow 0,1$. In $d>2$ dimensions, where sectors come from a compact group, the infinite index corresponds to Lie groups. This case was analyzed in \cite{Casini:2019kex} in terms of the mutual information and extended in \cite{Casini:2019nmu} (appendix B) to other Renyi entropies. It was found that the mutual information difference between the charged (complete) theory ${\cal C}$ and its neutral subtheory  ${\cal T}$ diverges logarithmically for any $d>2$ as 
\be
I_1^{\cal C}-I_1^{\cal T}\simeq -\frac{{\cal G} \,(d-2)}{4}\, \log \left(1-x\right)\;.
\ee
This is for disjoint regions with boundary given by two concentric spheres of radius $R\pm\epsilon/2$, in the limit of small separation $\epsilon\ll R$. The cross ratio is $x=(4 R^2-\epsilon^2)/(4 R^2)\simeq 1$. ${\cal G}$ is the dimension of the Lie algebra. This again has purely kinematical information about the symmetry of the model.  
This term is a negative contribution to the orbifold model with respect to the mutual information of the complete one, analogous to the $-\log |G|$ contribution for the finite group case. 
For two dimensional CFT's, the analogous formula is
\be
I_1^{\cal C}-I_1^{\cal T}\simeq -U_1^{\cal T}\simeq =\frac{{\cal G}}{2}\, \log( -\log (1-x))\,.
\label{dyy}
\ee
The same results are expected to be independent of the Renyi index (see next section). These results follow from the statistics of sharp twists for a continuous group. The expectation values of the twists form a Gaussian distribution around the identity, replacing (\ref{ty}). In particular, for $d=2$ CFT's, the twist can be constructed from a current which is a Gaussian field with dimension $1$. The logarithm of the expectation value of the projector to the neutral algebra gives (\ref{dyy}) \cite{Casini:2019kex,Casini:2019nmu}. 

From (\ref{mutualr2}), for $n=2$ this behavior implies 
\be
\log \frac{Z_{\cal T}(\beta)}{Z_{\cal C}(\beta)}=\log  \langle P\rangle_\beta\sim 
\frac{{\cal G}}{2}\,\log(\beta)\,,\label{ppa}
\ee
for the torus partition functions in the limit of large temperature $\beta=2\pi l\rightarrow 0$. Again, this can be obtained from the behavior of the twist expectation values at large temperature. For $\tau=e^{i \,a_i\,L_i}$, with $L_i$ the Lie algebra base, we have
\be
\langle \tau\rangle\sim e^{-k\,\frac{|\vec{a}|^2}{l}}\,,
\ee
with $k$ a constant. The twists with non vanishing probability are highly concentrated around the identity and the eventual non abelianity is not relevant in this limit. With this expression, one can also compute the expectation values of the projectors on the different representations, namely \ref{lab}, using the continuous version of \ref{proj1}, see \cite{Casini:2019kex}. In higher dimensions the analogous to (\ref{ppa}) for the thermal partition function in a sphere is \cite{kang2023universal} 
\be
\log \frac{Z_{\cal T}(\beta)}{Z_{\cal C}(\beta)}=\log  \langle P\rangle_\beta\sim 
\frac{{\cal G} \,(d-1)}{2}\,\log(\beta)\,.\label{ppax}
\ee

A particular case where the exact result for the mutual Renyi entropies is known for any Renyi index $n$ and cross ratio is the case of a free chiral scalar, that has
\be
U_n(x)=\frac{i\, n}{2(n-1)}\int_0^{+\infty} ds\,(\coth(\pi s n)-\coth(\pi s))\, \log\left(\frac{{}_2F_1(1+is, -is, 1, x)}{{}_2F_1(1 -is, is,1,x)}\right)\,.\label{red}
\ee
This was obtained from exact diagonlaization of the density matrix in \cite{Arias:2018tmw}.
The chiral scalar torus partition function is $Z(i l)=\eta(i l)^{-1}$, where $\eta$ is the Dedekind function. The identification of $I_2(x)$ obtained from (\ref{mutualr2}) using this torus partition function, and the formula for $I_2(x)$ using (\ref{red}), gives a quite obscure identity
\be
\eta(i l)= \exp{\left[\frac{1}{12}\log\left(\frac{x(1 - x)}{2^4}\right)+i\int_0^\infty ds\,\text{csch}(2 \pi s)\log\left(\frac{{}_2F_1(1+is, -is, 1, x)}{{}_2F_1(1 -is, is,1,x)}\right)\right] }\,,
\ee
where $x$ and $l$ are related by (\ref{tres}). We have checked numerically that the identity holds. 
The chiral scalar is the fix point algebra of the chiral complex free fermion under the $U(1)$ charge symmetry. This latter model is complete. We have from the limit of (\ref{red}) \cite{Arias:2018tmw}
\be
U_n(x)\sim -\frac{1}{2}\, \log(-\log(1-x))\,,\quad x\rightarrow  1\,.
\ee
This is consistent with (\ref{dyy}) for a $U(1)$ group. It also shows a limit behavior independent of the Renyi index $n$.  

For the Virasoro net $I_2(x)$ can be computed for any central charge $c\le 1$, by using the Virasoro character of the identity and (\ref{mutualr2}). For $c< 1$ the Virasoro net has finite index. The unitary models with $c<1$ have central charge 
 \be 
c=1-\frac{6}{m(m+1)}
\ee
with $m=3,4,\cdots$. The index for the chiral Virasoro net can be obtained from the known value of $S_{00}$ for the chiral characters \cite{francesco2012conformal}
\be 
\mu(m)=\frac{m(m+1)}{8}\left[\sin\left(\frac{m+1}{m}\pi\right)\sin\left(\frac{m}{m+1}\pi\right)\right]^{-2}\,.
\ee
This diverges in the limit $c\rightarrow 1, m\rightarrow \infty$. This formula appeared previously in \cite{Kawahigashi:2002px}.

For $c\ge 1$ the character of the identity is \cite{kac1979contravariant}
\be
\chi_0(l)=\frac{(1-q) \, q^{(1-c)/24}}{\eta(\tau)}\,,\quad \,\quad q=e^{-2\pi l}\,.
\ee
For $c=1$, using the property $\eta(-1/\tau)=\sqrt{-i \,\tau}\,\eta(\tau)$, $\tau=i l$, we get for the asymmetry
\be
A_2(l)=\log\left(\frac{\chi_1(l)}{\chi_1(1/l)}\right)=\frac{1}{2} \log(l)+\log\left(\frac{1-e^{- 2\pi l}}{1- e^{-2\pi/l}}\right)\,.
\ee
In the limit $l\rightarrow 0, x\rightarrow 1$, using (\ref{29}) we get
\be
U_2(x)\sim -\frac{3}{2}\, \log(-\log(1-x))\,.
\ee
Rehren in \cite{rehren1994new} showed that the Virasoro net for $c=1$ is the fix point algebra of the chiral $\hat{su}(2)_1$ current algebra at level one, under the action of the global $SU(2)$ symmetry. This current algebra has finite index. Then, this result for $U_2(x)$ is consistent with having a fix point algebra under the action of a compact continuous group with three generators, and with the generic formula \ref{dyy}.

For the Virasoro net with $c>1$ the torus partition function does not have the usual leading asymptotics. At large temperatures it scales as $\log Z\sim \frac{\pi}{12\, l}$ and it is not proportional to the central charge,  instead of the standard Cardy behavior $\log Z\sim \frac{\pi c}{12\,l}$. This is natural since there is only the stress tensor field in the algebra. If we naively use the formula (\ref{mutualr2}) in this case we get a larger divergence for $x\rightarrow 1$
\be
U_2(x)\sim  \frac{(c-1)}{12} \log(1-x) \,.
\ee
This cannot be obtained with a compact Lie group symmetry, and a larger divergence than a double logarithmic one is indeed expected. However, the problem with this formula is rather that for large $c$ one would not expect entropies scaling with $c$ in this case. This is also suggested by holography. In the following section, we understand why the formula (\ref{mutualr2}) is not directly applicable to this case.

\subsection{Revisiting the calculation of  \texorpdfstring{$I_2(x)$}{Lg} and other Renyi entropies}
\label{compu}

Now we describe the problem of computing these quantum information quantities in more generality in terms of partition functions. This will allow us to understand the case of higher Renyi entropies. It will also allow us to clarify the scope of the previous computations. In particular, they cannot be directly applied to the Virasoro net for $c>1$.

We assume the partition functions of the model can be obtained by modifying the ones of a complete theory $\mathcal{C}$. We can picture these later partition functions as ordinary path integrals with rotational invariant Lagrangians in an arbitrary $d=2$ manifold $M$. We can use these partition functions to analyze submodels by inserting appropriate projection operators along the manifold. These projections operators commute with the stress tensor and are topological line operators. This means they are lines that can be deformed without changing the partition function as long as we do not cross any other insertions. The projector lines can be open or closed.

Denote by $P_\gamma$ the projector from $\mathcal{C}$ to a subtheory $\mathcal{T}$. Then, according to section \ref{drd}, the path integral on a manifold with the insertion of a single projector $P_\gamma$ on an open curve $\gamma$ is given by
\be
\frac{Z_\gamma(M)}{Z(M)}=\lambda^{-1}\,,
\ee
where $\lambda$ is the global Jones index associated with the inclusion $\mathcal{T}\subset {\cal C}$.\footnote{For open defects this is universal because we are assuming sharp regularizations at the edges. Then the partition function, the expectation value of the projector, is completely controlled by the strong universal fluctuations on the edges.} On the other hand, when the projector line is a closed loop, the expectation value will depend on the specific theory and homotopy class of the loop in the manifold. If the path is contractible it can be assimilated to the identity. Similarly, as the line is a projector, the insertion of two closed lines that are homotopic to each other is equivalent to the insertion of a single one. Another property is that if we have a sphere with $n$ punctures and a projector going around each of the punctures, the partition function is the same as one with projectors on only $n-1$ of the punctures. The $n-1$ projectors are enough to ensure that no charges can appear in the remaining boundary. 

Certain limits of homotopically non trivial loops $\gamma$ can be understood. If the size of $\gamma$ goes to zero, only the identity can propagate in the direction perpendicular to the loop, and the loop does not contribute. The opposite limit in which the loop stretches a long distance or the transversal direction pinch to zero size is analogous to the high temperature limit and we have an $\lambda^{-1}$ contribution.

Now, if we have two open lines $\gamma_1$, $\gamma_2$, and take the limit in which the two end-points coalesce to get a closed curve $\gamma=\gamma_1\cup \gamma_2$, it is expected that the OPE of the product of operators should give the projector of the combined line\footnote{ Note this equation does not require that each of the projectors can be written as a product of local field operators at the edge. In fact, one can see examples \cite{Benedetti:2022zbb} in which the projectors can be non-local and still we expect such a relation to hold. Such relation is a particular form of ``strong additivity'', as we comment below.}
\be
P_{\gamma_1} P_{\gamma_2}\rightarrow  P_{\gamma}\,.\label{ffff}
\ee
The idea is that a charge-less operator in the surface $\gamma$ could contain opposite charges in $\gamma_1,\gamma_2$. But these can be approximated by charge-less operators in $\gamma_1, \gamma_2$ by adding opposite charges near the coalescing boundaries of both.  In (\ref{ffff}) a factor is expected that depends on the UV nature of the regularization of the projectors and the microscopic distance between the boundaries. These factors will not matter in what follows since we will be evaluating mutual informations that are combinations between different partition functions. These combinations eliminate these factors.

\begin{figure}[t]
    \centering
    \vspace{1cm}
    \includegraphics[width=0.8 \textwidth]{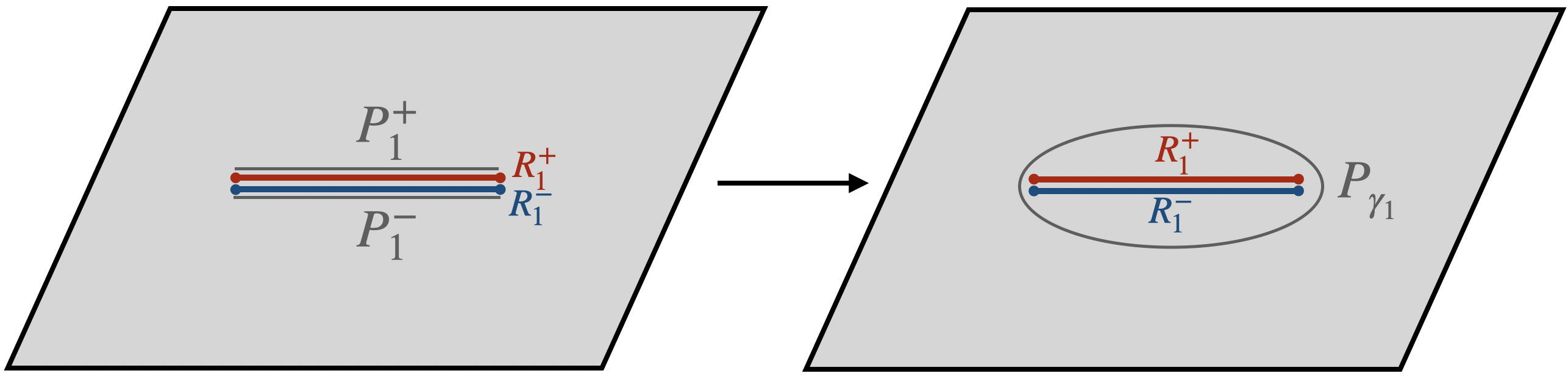}
     \caption{Projector $P_1^\pm$ lines on the two boundaries $R_1^\pm$  of an open cut along the interval $R_1$ coalesce into the protector $P_{\gamma_1}$ on the closed curved $\gamma_1 $ surrounding the interval. }
    \label{Proyectors1}
\end{figure}

Another form of understanding (\ref{ffff}) is the following. Suppose we have the theory in a circle $\gamma$ and divide the circle in two segments $\gamma_1,\gamma_2$. Equation (\ref{ffff}) follows from the stronger statement that
\be
{\cal A}(\gamma_1)\vee {\cal A}(\gamma_2)={\cal A}(\gamma)\,.
\ee
This property is called strong additivity and it is quite generally valid.\footnote{Here is it connected to (\ref{ffff}) assuming it is valid for the complete theory $\mathcal{C}$.} In particular it is valid for rational nets  \cite{kawahigashi2001multi,longo2004topological}. These are nets with finite global index and the split property.\footnote{The split property is however automatic for diffeomorphism invariant nets in the circle \cite{morinelli2018conformal}.}  However, quite insightfully, strong additivity is not valid for the Virasoro net with $c>1$ \cite{buchholz1990haag}. In this case, one can think heuristically that there are too many non local sectors in $\gamma$ to be generated by the algebras in $\gamma_1,\gamma_2$. This prevents the previous and following reasonings from being applied to this case and gives an explanation of the seemingly paradoxical result encountered at the end of the previous section. We expect the scope of the following arguments restricts to strongly additive nets.  
\begin{figure}[t]
    \centering
    \vspace{1cm}
    \includegraphics[width=0.8 \textwidth]{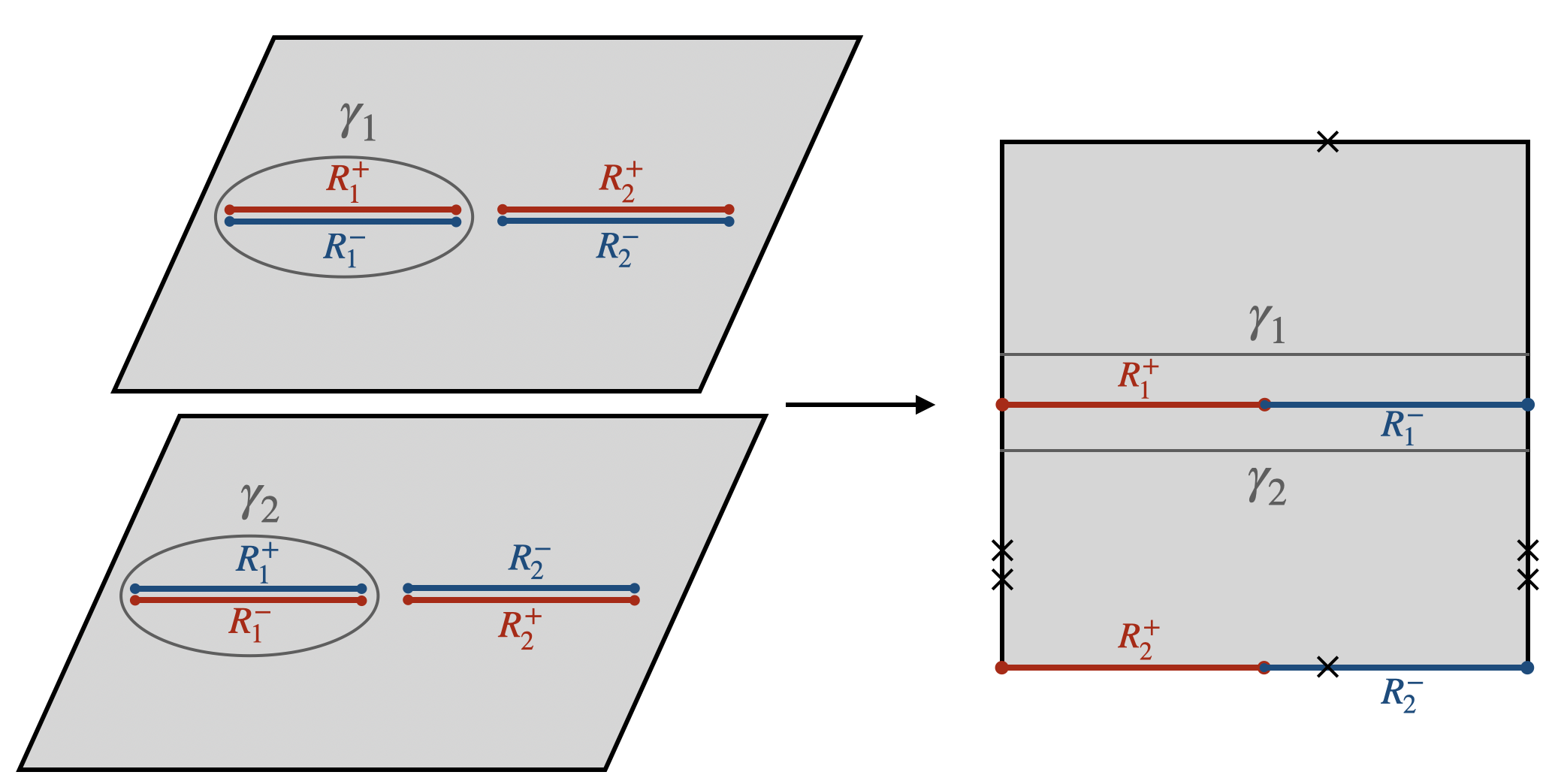}
     \caption{The $n=2$ Renyi entropy for two intervals can be computed as a partition function on the two-copy plane gluing $R^\pm_{1,2}$ cyclically as shown in the figure of the left. Inserting the corresponding projectors over $\gamma_1$ and $\gamma_2$ computes the Renyi entropy for the submodel. The resulting configuration can be mapped to a genus $g=1$ torus as in the right-hand side of the figure. It is a rectangle with opposite sides identified. Note $\gamma_1,\gamma_2$ are homotopical curves.}
    \label{Proyectors2}
\end{figure}

Now we apply these tools for the computation of the Renyi mutual informations. The vacuum density matrix for the theory $\mathcal{C}$ on a region with several intervals $R_i$ is computed with the partition function in the plane with open cuts at $R_i$. For the submodel $\mathcal{T}$ we have to project to the neutral Hilbert space in each of the intervals, inserting projector lines along the two open boundaries $R_i^\pm$ of each interval. See figure \ref{Proyectors1}.  The density matrix needs regularization at the end points of the intervals. A natural way to deal with the regularization and the projectors is to displace them a bit out of the intervals cuts. We then coalesce them as a closed curve surrounding the intervals. See figure \ref{Proyectors1}. We then have a projector curve $\gamma_i$ encircling each of the intervals $R_i$.

For a single interval region, the projector can be eliminated because it is contractible towards infinity since the plane is conformally equivalent to a sphere. This is not the case for multi-interval regions. We focus on the case of two intervals. In this case, one of the two projectors in the plane around each interval can be eliminated because the two projectors can be deformed to each other.  Computing the Renyi entropies for two intervals we have to glue $n$ copies of the cut plane in cyclic order along the boundaries. This leaves us with $n$ projectors around the $n$ copies of the first interval. For $n=2$ we have genus $1$ and the manifold can be transformed conformally to a torus. The result was described in section \ref{secII}. It is not difficult to follow the path of the projectors around the first interval to see that they are equivalent projectors along the spatial circle of the torus. See figure \ref{Proyectors2}. So the partition function $Z[il]$ that appears in the formula of $I_2(x)$ is indeed the thermal partition function of the model ${\cal T}$. When $R_1,R_2$ are near to each other the projector around $R_1$ gets squeezed. In this case, we are in the large temperature limit. For finite index the difference in  the Renyi mutual information between the models reduces to the index, 
\be
\lim_{x\rightarrow 1}\,(I^{{\cal C}}_2(x)-I^{\cal T}_2(x))=-\lim_{x\rightarrow 1}U^{\cal T}_2(x)=\log \lambda\,.
\ee
This comes from the contribution of the projector, since all other contributions, including cutoff and the universal Liouville contribution of the conformal transformation to the torus, cancel. In the same way, for the case of a compact Lie group symmetry the formula (\ref{dyy}) applies for $I_2$. 

\begin{figure}[t]
    \centering
    \vspace{1cm}
    \includegraphics[width=0.35 \textwidth]{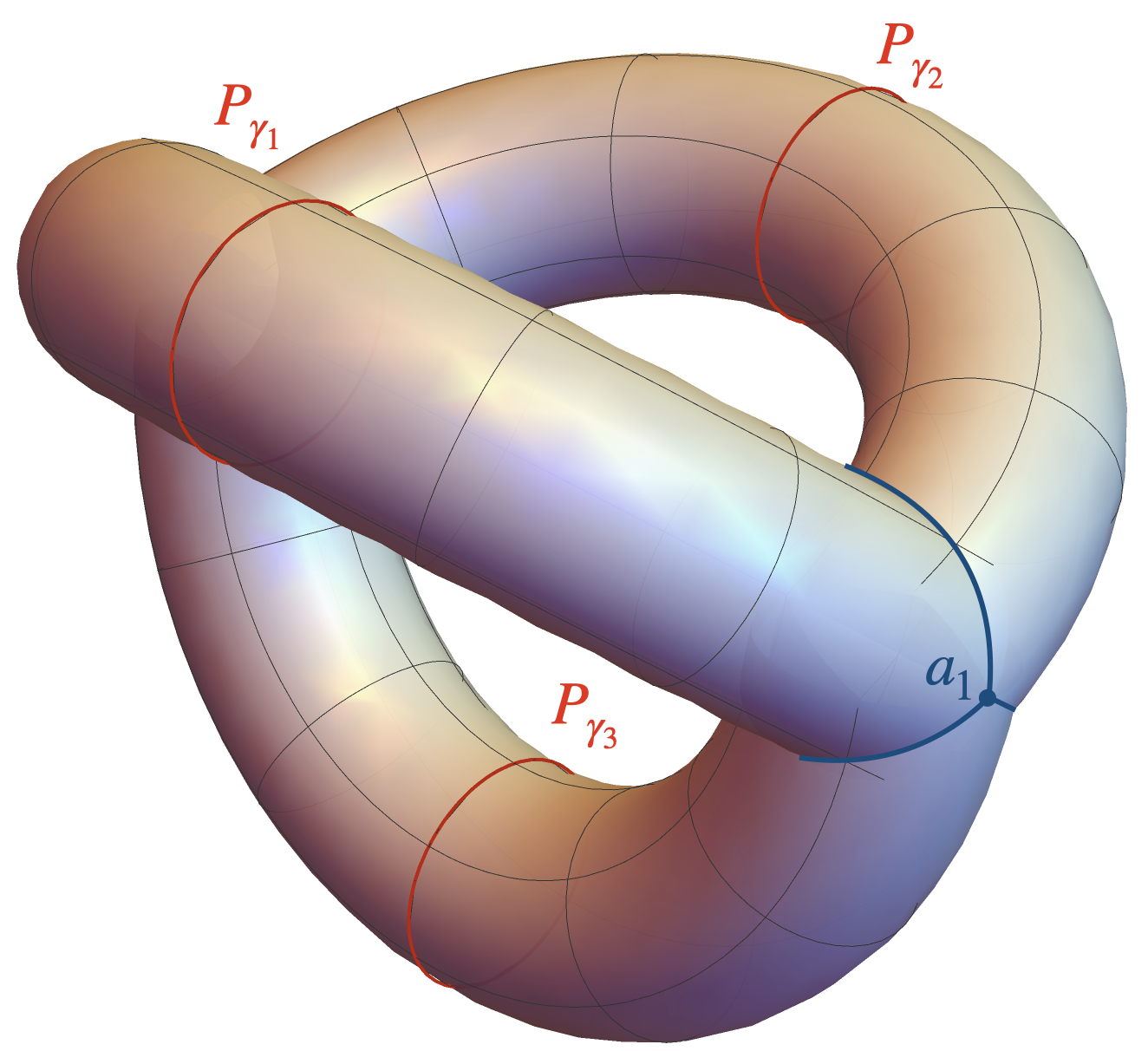}
     \caption{The $n=3$ Renyi entropy calculation for two intervals can be mapped to a genus $2$ surface with three handles and a cyclic $Z_3$ symmetry between them. When pursuing the calculation in a submodel, we have three projectors $P_{\gamma_1}$, $P_{\gamma_2}$ and  $P_{\gamma_3}$ appearing as circles around each one of the handles. One of the projectors is always redundant. The marked point $a_1$ is one of the endpoints of the interval $R_1$, and is a branching point of the original geometry. The marked three lines starting at $a_1$ are the three copies of the interval $R_1$, that connect $a_1$ with the other endpoint $b_1$ of $R_1$ (not shown). }
    \label{Proyectors3}
\end{figure}

It is not difficult to understand the generalization for higher $n$. In this case, the manifold has genus $n-1$ and can be transformed to a topology that is the one of $n$ handles joined together at the two extremes. At these extremes, there are $n$ lines joining two points. The lines are the $n$ copies of the intervals $R_1,R_2$, and the points are the end-points of the intervals that are branching points of the original geometry. See figure \ref{Proyectors3} illustrating the case $n=3$. The geometry has cyclic $Z_n$ permutation symmetry. Each projector is now a circle around one of the $n$ handles. One of the projectors is redundant. In the limit $x\rightarrow 1$ the projectors get squeezed and are in the ``high temperature limit''. As a result, and because of the factor $(n-1)^{-1}$ in the definition of the Renyi entropy, it follows that the index can be read off  from the limit of any of the Renyi mutual information     
\be
\lim_{x\rightarrow 1}\,U^{\cal T}_n(x)=-\log \lambda=-\frac{1}{2}\log \mu=-\frac{1}{2}\log \left(\sum_{r\in DHR} d_r^2\right)\,.
\ee
Hence we have a ``flat spectrum'', independent of Renyi index, for this topological contribution. This terminology comes from the analogy with the Renyi entropies of a finite dimensional density matrix proportional to the identity. This also has constant Renyi entropies $\log d$, independent of $n$, where $d$ is the Hilbert space dimension. Here we have an analogous scenario for the replica modular asymmetry which is proportional to the logarithm of the total quantum dimension of a generic modular tensor category.

It is known that there is Haag duality for two intervals if and only if there is Haag duality for any number of intervals \cite{kawahigashi2001multi}. This can be simply understood from the fact that Haag duality violations are produced by combinations of charged point-like fields in two or more intervals. The global index of the inclusion of the algebra ${\cal A}(I_1,\cdots ,I_m)\subset\hat{{\cal A}}(I_1,\cdots ,I_m)$ is $\mu^{m-1}$. This index could also be recovered from Renyi entropies and higher genus partition functions in the limit in which all intervals touch each other and charge-anticharge operators get maximal expectation values.  

\section{Modular invariance, completeness and superselection   sectors}\label{secIV}

Previously we have analyzed the direct relation between modular invariance and completeness using Renyi entropies. Completeness can be formalized as Haag duality for general regions and then violations of Haag duality, measured by e.g. Renyi entropies and order parameters, can be related to departures from modular invariance. We now present a complete mathematical proof that ultimately stems from the DHR theory of superselection sectors \cite{Doplicher:1971wk,Doplicher:1973at,Doplicher:1990pn}, generalized to two dimensional theories and subfactors in \cite{cmp/1104179464,FRS,Longo:1994xe}.

As mentioned in the introduction, that $S$ modular symmetry is related to the absence of superselection sectors was conjectured by Rehren \cite{R6}. The context at the time was the analysis and potential classification of local 2-dimensional conformal quantum field
theory $\mathcal{T}$ which irreducible extends a given pair of chiral
theories $\mathcal{A}=\mathcal{A}_L\otimes \mathcal{A}_R$.  More precisely, the mathematical structure is an irreducible inclusion of nets,
$\mathcal{A}_L(J_+)\otimes \mathcal{A}_R(J_-)\subset
\mathcal{T}(O)$, where $J_+, J_-$ are light ray intervals and $O$ is a double cone
$J_+\times J_-$. Although not standard, notice that in principle the chiral algebras $\mathcal{A}_L$ and $\mathcal{A}_R$ can be distinct. We only require the extension to be local.

The chiral algebra $\mathcal{A}$ typically has superselection sectors associated with it. Following DHR, see \cite{haag2012local} for an introduction to the subject, the localizable sectors can be encoded in the vacuum sector via endomorphisms $\rho_r$ of the chiral algebra. These endomorphisms form a category, which in $d=2$ is a braided modular category. In this context, we have further information. Since we have an inclusion of algebras $\mathcal{A}_L(J_+)\otimes \mathcal{A}_R(J_-)\subset
\mathcal{T}(O)$, we have a canonical DHR endomorphism $\th:\mathcal{A}_L(J_+)\otimes \mathcal{A}_R(J_-)\rightarrow \th(\mathcal{A}_L(J_+)\otimes \mathcal{A}_R(J_-))\subset\mathcal{A} (\mathcal{O})$ associated with the inclusion \cite{L11}. Denoting by $j_{\mathcal{A}}$ and $j_{\mathcal{T}}$ the modular conjugations associated with $\mathcal{A}(\mathcal{O})$ and $\mathcal{T}(\mathcal{O})$ respectively, the canonical endomorphism is defined as
\be 
\th(A)=j_{\mathcal{A}}j_{\mathcal{T}}\,A\,j_{\mathcal{T}}j_{\mathcal{A}}\,\quad \, A\in \mathcal{A}(\mathcal{O})\;.
\ee
This endomorphism is generically reducible. Since the chiral algebra is a tensor product of left/right algebras, the irreducible endomorphisms $\rho$ are tensor products of irreducible endomorphisms $\alpha^L_r$ and $\alpha^R_s$ of left/right algebras, namely $\rho_{rs}= \a_r^L\otimes\a_s^R$. We have a general expression of the form
\be
\th=\bigoplus_{rs} M_{rs} \a_r^L\otimes\a_s^R
\ee
This parallels, in the neutral sector of the chiral algebra, the particular combination of left and right superselector sectors that naturally appears in the extended theory $\mathcal{T}$. More concretely, the present matrix $M$ is the same as the one appearing in the partition function in \ref{partM}. We thus have $M_{00}=1$
and $M_{rs}\in\mathbb{N}$. In the mathematical literature, the matrix $M$ is called a coupling matrix (and generally denoted $Z$).  The two nets $\mathcal{A}_L$ and $\mathcal{A}_R$ define $S$ and $T$ matrices,
$S_L$, $T_L$, $S_R$, $S_R$, respectively, as in \cite{R1}. We are interested in the case where the $S$-matrices are
invertible. By the results in \cite{kawahigashi2001multi}, this invertibility,
which is called non-degeneracy of the braiding,
 holds if the nets are completely rational in the
sense of \cite{kawahigashi2001multi}.  Then Rehren considered when the
following two intertwining relations hold.
\begin{equation}
\label{stat-symm}
T_L M = M T_R,\qquad\qquad S_L M= M S_R.
\end{equation}
 This condition is the usual
modular invariance.    
 Rehren then conjectures that a possible criterium
to enforce the intertwining property is that the local 2D theory $\mathcal{C}$ does not possess nontrivial superselection sectors.
We now prove this conjecture, namely that the triviality of the superselection structures is indeed sufficient (and necessary) for the intertwining property (\ref{stat-symm}), when $\mathcal{A}_L$ and
$\mathcal{A}_R$ are completely rational.  We will also
prove that this condition is equivalent to ``maximality'' of the
extension $\mathcal{C}$ (namely the model is complete) if a certain natural symmetry holds for
$\mathcal{A}_L$ and $\mathcal{A}_R$. In this case, we will have in particular $\mathcal{A}_L=\mathcal{A}_R$ as a part
of the assumptions.

\subsection{Preliminaries}

Before presenting the proof we notice that the results for two interval Jones index (the index $\mu$ considered above) derived in \cite{kawahigashi2001multi} are also valid
for 2-dimensional nets $\mathcal{C}(\mathcal{O})$, where $O$ is a two-dimensional double cone.

In particular, we have an analogue of \cite[Corollary 32]{kawahigashi2001multi}, which reads

\begin{proposition}
\label{mu}
For a 2-dimensional completely rational net $\mathcal{C}$, the following statements are equivalent.
\begin{enumerate}
\item The net $\mathcal{C}$ has no non-trivial sector with finite dimension.
\item The net $\mathcal{C}$ has no non-trivial sector (with finite of
infinite dimension).
\item The two interval Jones index $\mu_{\mathcal{C}}$ is 1.
\end{enumerate}
\end{proposition}

This proposition then relates the absence of DHR superselection sectors of the net (a fair notion of completeness), with the validity of Haag duality for two interval regions.\footnote{One can show that triviality of the two interval inclusion of algebras implies triviality of the n-interval inclusion of algebras \cite{kawahigashi2001multi}.} 

\subsection{Modular invariance and triviality of the superselection
structure}

To prove Rehren's conjecture we now add a further equivalent statement to the previous proposition, namely invariance under $S$ and $T$ transformations.\footnote{In fact we remind that invariance under $T$ transformations is ensured by locality of the net $\mathcal{T}$, as discussed above.} We arrive at the following theorem:

\begin{theorem}
\label{modular}
Under the above conditions, the following are equivalent.
\begin{enumerate}
\item The net $\mathcal{C}$ has only the trivial superselection sector.\label{tec1}
\item The $\mu$-index $\mu_{\mathcal{C}}$ is 1.\label{tec2}
\item The matrix $M$ has the intertwining property 
(\ref{stat-symm}),
$$T_L M = M T_R,\qquad\qquad S_L M= M S_R.$$ \label{tec3}
\end{enumerate}
\end{theorem}

\begin{proof}
We first note that whenever we have an inclusion of the form  $\mathcal{A}_L(J_+)\otimes \mathcal{A}_R(J_-)\subset
\mathcal{C}(O)$, we also have an intermediate inclusion of algebras of the form
\begin{equation}
\label{interm}
\mathcal{A}_L(J_+)\otimes \mathcal{A}_R(J_-)\subset \mathcal{A}_L^{\max}(J_+)\otimes \mathcal{A}_R^{\max}(J_-)
\subset \mathcal{C}(O)\;,
\end{equation}
see \cite{R4}. Equivalently, for any inclusion of the original form there is always a maximal chiral extension providing an intermediate inclusion. Denoting by $B_L$ and $B_R$ the branching (generically rectangular) matrices relating the original chiral algebra with the maximal ones, we can write the coupling matrix as
\be
M=B^t_L M^{\max} B_R\;,
\ee
where $M^{\max}$ is the coupling matrix for $ \mathcal{A}_L^{\max}(J_+)\otimes\mathcal{A}_R^{\max}(J_-)
\subset \mathcal{C}(O)$. Note also that the index of the inclusion $\mathcal{A}_L(J_+)\otimes \mathcal{A}_R(J_-)\subset
\mathcal{C}(O)$ is automatically finite, because $\mathcal{A}_L$ and
$\mathcal{A}_R$ are completely rational
and thus \cite[Proposition 2.3]{Kawahigashi:2002px} applies. 
Equivalence between (\ref{tec1}) and (\ref{tec2})
has been already proved in Proposition \ref{mu}.

Now we prove that (\ref{tec2}) implies (\ref{tec3}).
First, Corollary 3.2 in \cite{R6} says that $M$ intertwines
the diagonal matrices of statistics phases. Let $\Delta_L$, $\Delta_R$, $\Delta_L^{\max}$, $\Delta_R^{\max}$,
and $\Delta_{\mathcal{C}}$ be the systems of
irreducible DHR endomorphisms for $\mathcal{A}_L$, $\mathcal{A}_R$, $\mathcal{A}_L^{\max}$,
$\mathcal{A}_R^{\max}$ and $\mathcal{C}$, respectively.   We denote their two interval indices
indices, namely the square sums of the statistical dimensions in
each of the systems,  by $\mu_L$, $\mu_R$, $\mu_L^{\max}$, $\mu_R^{\max}$, and
$\mu_{\mathcal{C}}$, respectively.
By Rehren's result \cite[Corollary 3.5]{R4},
the dual canonical endomorphism for the subfactor
$\mathcal{A}_L^{\max}(J_+)\otimes \mathcal{A}_R^{\max} (J_-)\subset \mathcal{C}(O)$ decomposes as
\be
\label{deco}
\theta_{\max}=\bigoplus_r \sigma_r^L \otimes \tilde\sigma_r^R\;.
\ee
Here $\{\sigma_r^L\}\subset \Delta_L^{{\max},0}$ and 
$\{\tilde\sigma_r^R\}\subset \Delta_R^{{\max},0}$ make
closed subsystems of irreducible endomorphisms with $\Delta_L^{{\max},0}\subset \Delta_L^{\max}$
and $\Delta_R^{{\max},0}\subset \Delta_R^{\max}$. Also, the map $\sigma_r\mapsto \tilde\sigma_r$ gives a fusion
rule isomorphism.  Therefore the indices of $\Delta_L^{{\max},0}$
and $\Delta_R^{{\max},0}$ are the same and we denote them by
$\mu_0$. We can then compute the index of the subfactor
$\mathcal{A}_L^{\max}(J_+)\otimes \mathcal{A}_R^{\max} (J_-)\subset \mathcal{C}(O)$. In \cite{Longo:1994xe} it was shown that such index equals the dimension of the canonical endomorphism. Since $d_{\sigma_r}=d_{\tilde\sigma_r}$, this dimension is then equal to $\mu_0$. We now use again Proposition 24 in \cite{kawahigashi2001multi}, which was used above in eq. \ref{twointglob}. There it was shown that for two models $\mathcal{T}\subset \mathcal{C}$, the global (two interval) indices of the two models satisfy the following relation
\be
\mu_{\mathcal{T}}= \lambda^2\,\mu_{\mathcal{C}}\,.
\ee
In the present scenario $\mathcal{T}\rightarrow \mathcal{A}_L^{\max}(J_+)\otimes \mathcal{A}_R^{\max} (J_-)$ and we obtain
\be
\mu^{\max}_L\, \mu^{\max}_R= \mu_{\mathcal{C}}\, \mu_0^2\;.
\ee
Since we seek to prove that (\ref{tec2}), namely $\mu_{\mathcal{C}}=1$, implies (\ref{tec3}) we arrive at $\mu^{\max}_L\, \mu^{\max}_R= \mu_0^2$. We also know that $\mu_0\le\mu_L^{\max}$ and $\mu_0\le\mu_R^{\max}$ due to
$\Delta_L^{{\max},0}\subset \Delta_L^{\max}$
and $\Delta_R^{{\max},0}\subset \Delta_R^{\max}$. Then we have
$\mu_0=\mu_L^{\max}=\mu_R^{\max}$, which in turn implies
$\Delta_L^{{\max},0}=\Delta_L^{\max}$
and $\Delta_R^{{\max},0}=\Delta_R^{\max}$.

We conclude that if the extended model $\mathcal{C}$ has trivial index $\mu_{\mathcal{C}}=1$, then the canonical endomorphism contains all sectors in $\mathcal{A}^{\max}$. Then Corollary 3.8 in \cite{R6} gives the desired conclusion based on 
$S_L^{\max} M^{\max}=M^{\max} S_R^{\max}$ and proven by
\cite[Theorem 6.5]{BE4}.

Now we prove that (\ref{tec3}) implies (\ref{tec2}).
First note that
\be
M^{\max}=(M^{\max}_{\tau\sigma})_{\tau\in\Delta_L^{\max},
\sigma\in\Delta_R^{\max}}
\ee
contains a permutation matrix
\be
(M^{\max}_{\tau\sigma})_{\tau\in\Delta_L^{{\max},0},
\sigma\in\Delta_R^{{\max},0}}=
(\delta_{\pi(\tau),\sigma})\;.
\ee
as a submatrix, where $\pi$ gives the permutation, and the other parts of the matrix $M^{\max}$ are 0. These arises from \cite[Corollary 3.5]{R4}. Since we have the decomposition $M=B^t_L M^{\max} B_R$, assuming (\ref{tec3}) we have
\begin{equation}
\label{intertwine}
B^t_L(S_L^{\max} M^{\max}-M^{\max} S_R^{\max})B_R=0
\end{equation}
by \cite[Theorem 6.5]{BE4}.  Note that for the branching matrix
$B_R=(b_{R,\sigma\rho})$ with $\sigma\in \Delta_R^{\max}$,
$\rho\in\Delta_R$, we have $b_{R,\sigma\rho}=\lan \a_\rho,\sigma\ran$.
Thus for $\rho=0$, for $\rho$ equal to the identity, we have
$b_{R,\sigma 0}=\delta_{\sigma 0}$.  A similar identity holds for
$B_L$.  Then the $(0,0)$-entry of the equation (\ref{intertwine})
gives $S_{L,00}=S_{R,00}$.  Corollary 3.6 in \cite{R5} now shows
that the entries of the matrix  
\be
R=(R_{\tau\sigma})=S_L^{\max} M^{\max}-M^{\max} S_R^{\max}
\ee
are given as follows:
\begin{equation}
\label{eq1}
R_{\tau\sigma}=\left\{
\begin{array}{ll}
0, &\hbox{if $\tau\in\Delta_L^{{\max},0}$, $\sigma\in\Delta_R^{{\max},0}$,}\\
S^{\max}_{L,\tau,\pi^{-1}(\sigma)}, &
\hbox{if $\tau\notin\Delta_L^{{\max},0}$, $\sigma\in\Delta_R^{{\max},0}$,}\\
-S^{\max}_{R,\pi(\tau),\sigma}, &
\hbox{if $\tau\in\Delta_L^{{\max},0}$, $\sigma\notin\Delta_R^{{\max},0}$,}\\
0, &\hbox{if $\tau\notin\Delta_L^{{\max},0}$,
$\sigma\notin\Delta_R^{{\max},0}$}.
\end{array}
\right.
\end{equation}
Then the $(0,\sigma)$-entry of the product $B^t_L R$ is now computed to be
\begin{equation}
\label{eq2}
(B^t_L R)_{0\sigma}=\left\{
\begin{array}{ll}
0, &\hbox{if $\sigma\in\Delta_R^{{\max},0}$,}\\
-S^{\max}_{R,0\sigma}, &
\hbox{if $\sigma\notin\Delta_R^{{\max},0}$.}
\end{array}
\right.
\end{equation}
Since we have $B^t_L R B_R=0$, the entries of $B_R$ is non-negative,
and $S^{\max}_{R,0\sigma} > 0$, we have 
$B_{R,\sigma\rho}=0$ for $\sigma\notin\Delta_R^{{\max},0}$ and all
$\rho$.  This means that we have no $\sigma\notin\Delta_R^{{\max},0}$ 
and thus $\Delta_R^{{\max},0}=\Delta_R^{\max}$.   Similarly, we have
$\Delta_L^{{\max},0}=\Delta_L^{\max}$.
By the above arguments
for $(2)\Rightarrow(3)$, we know that this implies $\mu_B=1$.
\end{proof}

\subsection{Modular invariance and maximality of the extension}

In this section we further show that if $\mathcal{A}_L$ and $\mathcal{A}_R$ have a certain
natural symmetry property, then the conditions in
Theorem \ref{modular} are also equivalent to the maximality
of the extension $\mathcal{C}$ of $\mathcal{A}_L\otimes \mathcal{A}_R$. We  study the
Longo-Rehren subfactors \cite{Longo:1994xe}
for this purpose.

We now assume that the 2-dimensional net $\mathcal{C}$ is
invariant under the reflection $x\mapsto -x$ in the
Minkowski spacetime coordinates $(x,t)$.
This assumption in particular implies that we
have a natural identification
$\mathcal{A}_L=\mathcal{A}_R$ of
the left and right chiral conformal nets.
For this reason, we now drop the superscripts
L and R in the above notation.  Then the 
decomposition (\ref{deco}) of 
the dual canonical endomorphism for the subfactor
$\mathcal{A}_L^{\max}(J_+)\otimes \mathcal{A}_R^{\max} (J_-)\subset \mathcal{C}(O)$ is now of the form
\be
\theta_{\max}=\bigoplus_r \sigma_r \otimes \sigma_r\;,
\ee
where the system $\{\sigma_r\}$ is a
closed subsystem of irreducible endomorphisms with $\Delta^{\max,0}\subset \Delta^{\max}$.
This gives an intermediate subfactor corresponding to
the dual canonical endomorphism
\be
\theta'=\bigoplus_{\sigma\in\Delta^{\max}}
\sigma \otimes \sigma\;,
\ee
arising from the Longo-Rehren subfactor.  If we have $\Delta^{{\max},0}\neq \Delta^{\max}$, this gives a proper extension of a 2-dimensional net $\mathcal{C}$. Noting that the ``only if'' part is trivial, we have proved the following

\begin{theorem}
\label{maximality}
Under the above conditions including the invariance
of the 2-dimensional net $\mathcal{C}$
under the reflection $x\mapsto -x$, 
we have the three equivalent conditions in
Theorem \ref{modular} if and only if the net
$\mathcal{C}$ is maximal with respect to 
irreducible inclusions of 2-dimensional nets.
\end{theorem}

\section{Completeness and thermal partition function in higher dimensions}\label{secV}

In this section, we make progress towards understanding how to detect failures of completeness from the CFT data in higher dimensions. We will only discuss the case of Haag duality for disjoint regions, or DHR completeness. This is the analog to the case studied above for $d=2$. In higher dimensions, the DHR theorem states that a non complete model $\mathcal{T}$ in this sense is the fix point algebra of a complete model $\mathcal{C}$ under some compact symmetry group $G$, namely it is always the case there is a group average $\epsilon:\mathcal{C}\rightarrow \mathcal{T}$. If the group is finite the global Jones index is the size of the group $|G|$.

Let us recall the salient features of the previous discussion to compare with a possible generalization to higher dimensions. In $d=2$ the partition function of the torus has the geometrical modular symmetry for complete theories. In this case, since the low temperature expansion (\ref{aba}) does not contain a constant term because of the discreteness of the spectrum, the same is implied for the high temperature expansion (\ref{aba1}). Then, any failure of completeness, that breaks the geometric modular symmetry, will give place to a subleading term in (\ref{aba1}). This subleading term allows to read off cleanly the global index of the model.  

In higher dimensions, the same type of geometric symmetries is expected for partition functions in multi-torus topologies $S^1\times\cdots \times S^1$. The index will make a contribution to these partition functions at high temperature, exactly as discussed in section \ref{secIII}. Unfortunately, these partition functions are not related to the CFT data in a simple way and it is not known how to compute them in a general case. 

On the other hand, the partition function in the manifold $S^{d-1}\times S^1$ is the thermal partition function for the theory on the cylinder. The Hamiltonian eigenvalues are directly related to the field scaling dimensions. However, there is no geometric symmetry in this higher dimensional case. This in turn does not allow to connect with the low temperature partition function. It then does not a priori forbids dynamical constant contributions in the free energy in the high temperature limit. These possible contributions in principle would not allow to isolate the index from the constant term. 

However, there is a general argument suggesting that in even spacetime dimensions the constant term is absent for complete theories. At large temperatures, the partition function is dominated by local high energy excitations whose statistics is independent for spatially separated points at distances larger than the thermal wavelength.  
Then, we expect an extensive leading scaling of the free energy $\log Z=c_{d-1} \, R^{d-1}\, T^{d-1}+\cdots$, where $c_{d-1}$ is an analog of the Stefan-Boltzmann constant, and where $R$ is the sphere radius. The subleading perturbative terms in the large temperature expansion are expected to be local corrections, proportional to polynomials on the curvature. There are also non perturbative corrections that are exponentially small at high temperature. This framework is called the high temperature effective action \cite{jensen2012towards,banerjee2012constraints,benjamin2024universal}. The curvature corrections introduce even powers $R^{2k}$ of the sphere radius, corresponding to powers $T^{d-1-2 k}$ in the large temperature expansion. For even dimensions, this does not contain a constant term. For example, in $d=4$, and setting $R=1$, we will have\footnote{There are particular ``pathological'' exceptions such as the free scalar. One way to think about this is that in the small $S^1$ size, the dimensionally reduced theory in $S^{d-1}$ is generically gapped. This fails for the scalar which contains a gapless sector \cite{benjamin2024universal,lei2024modularity}.}
\be
\log Z\simeq  c_{3}\, T^3+c_1\, T^1+c_{-1}\, T^{-1}+\cdots  \label{parame}
\ee
For the neutral subtheory $\mathcal{T}$, there are global constraints that go beyond this local effective description. Then, for such subtheory $\mathcal{T}$ we will expect instead
\be
\log Z\simeq  c_{3}\, T^3+c_1\, T^1-\log |G|+\cdots  \
\ee
For Lie group symmetry, the term (\ref{ppax}) logarithmic in the temperature is expected. And for different irreducible representations, the same behavior as in \ref{lab} is also expected. 

This neatly identifies the non completeness of the model from CFT data and, in this case, only the spectrum of conformal dimensions is needed. This solves the problem, at least for even dimensions.  In contrast, for odd dimensions it is not possible to avoid dynamical constant terms in the free energy $\log Z$, and a finite index seems difficult to isolate.

The density of operators $\rho(\Delta)$ as a function of scaling dimension is obtained from the partition function as a Laplace anti-transform. The large  $\Delta$ regime involves a saddle point approximation in the regime of large temperature. Then, expansions of $\rho(\Delta)$ for large $\Delta$ depend on the expansions of $\log Z(\beta)$ for small $\beta$. For finite index, the multiplicative  change in $Z$ leads to 
\be
\log \rho\simeq \log \rho_c -\log |G|\,,\label{doce}
\ee
where $\rho_c$ is the density for a complete theory. The constant term changes the functional form of the complete model and is recognizable by the knowledge of  $\rho$ itself. For example, in $d=2$, from the leading term including the  prefactor in Cardy formula for a modular invariant theory \cite{carlip2000logarithmic},  we get 
\be
\rho\simeq \frac{1}{\lambda}\,\left(\frac{c}{96 \Delta^3}\right)^{\frac{1}{4}}\, e^{2\pi \sqrt{\frac{c \,\Delta}{6}}}\, \left(1+{\cal O}(\Delta^{-1/2})\right)\,.
\ee
The index can be identified from the density of states at large $\Delta$ from the anomalous prefactor, or from the constant term in the expansion 
\be
\log \rho\simeq 2\pi \sqrt{\frac{c \,\Delta}{6}}-\frac{3}{4}\,\log \Delta +\frac{1}{4}\,\log (c/96)-\log \lambda+{\cal O}(\Delta^{-1/2})\,.
\ee
Expansions of the density of states for complete theories in higher dimensions are likewise parametrized by the coefficients of the expansion (\ref{parame}) of the high temperature partition function, see \cite{benjamin2024universal}. The order of the group $|G|$ of ``missing symmetries'' can be identified in the same way from the anomalous constant coefficient in (\ref{doce}). 

For Lie group symmetries we have a correction of the free energy that is logarithmical in the temperature (\ref{ppax}). This is not expected for complete models in odd or even dimensions and hence gives a way to identify the non completeness in any dimensions in this case. In terms of the density of states, the saddle point identification $\beta \sim \Delta^{-1/d}$, gives  
\be
\log \rho\simeq \log \rho_c -\frac{{\cal G} \,(d-1)}{2 \,d}\,\log(\Delta)\,.
\label{add}
\ee
 which also allows to detect the presence of an incomplete model from the anomalous coefficient of the logarithmic term. For example, for $d=3$  the logarithmic term for complete theories is universally $-4/3\, \log(\Delta)$ \cite{benjamin2024universal}, and a Lie group orbifold would add the second term in (\ref{add}).

As a final commentary, it is interesting to compare with the case of the entanglement entropy of a sphere. This is equivalent to the computation of the partition function in a $d$-dimensional Euclidean sphere \cite{Casini:2011kv}. Again, the index is expected to appear in the constant term of the entanglement entropy \cite{Casini:2019kex}. In odd dimensions, this will be mixed up with the constant $F$ term. In even dimensions, although no universal constant term is expected (as the thermal partition function above), there is a logarithmically divergent term due to the trace anomaly on the sphere. This logarithmic term is absent for $S^{d-1}\times S^1$ \cite{benjamin2024universal}. The uncertainty in the normalization of the cutoff in the logarithmic term then ruins the identification of the index in the constant term of the entanglement entropy. Still, such index appears clearly in the mutual information and related quantities, as discussed above.

\section{Discussion}\label{SecVI}

In this article we have exposed the direct connections between modular invariance, locality, and completeness. While T-duality is required for locality of the QFT, S-invariance is not. We have seen that S-invariance is a form of completeness of the QFT.  More precisely it is equivalent to the validity of Haag duality for generic regions, or to the absence of superselection sectors. From the perspective of Haag duality, it admits a natural generalization to higher dimensions. These connections have been explained using quantum information techniques, partition function with insertions of topological defects, and the DHR theory of superselection sectors. We now wish to discuss connections with other works and some open problems.

\textbf{Quantum Field Theory Axiomatics:} A general framework for a description of CFT's as relativistic quantum theories is given by the algebraic approach \cite{haag2012local}. This incorporates as basic elements the quantum operators and locality. Important advances in the understanding and classification of CFT's have been accomplished along this line of research, e.g \cite{Kawahigashi:2002px,Kawahigashi:2003gi}. A parallel framework in terms of partition functions and Hilbert spaces, as in the approach to topological field theories, has the advantage to put into first focus the geometrical aspects. Segal's axioms \cite{Segal1988} offer such a purely geometrical description incorporating an associative fusion algebra.  One of the conclusions of this paper is that a purely geometric description in terms of partition functions as in Segal's axioms must correspond only to complete relativistic QFT models. Non complete models contain additional elements on top of the geometry. But they are perfectly sound as relativistic QFT. Equivalently, although
associativity, closedness of the fusion algebra and the fact that three-point functions are enough to solve the CFT follow nicely and geometrically from Segal's axioms, this does not mean that all CFT's that satisfy those three properties satisfy Segal's axioms. It would be interesting to show that this completeness implicit assumption is the only or main difference between both axiomatics when referring to CFT's. 

\textbf{Partition functions and modular theory:} There is an understanding of generic partition functions in terms of the basic elements of the quantum theory. The simplest example is the torus partition function which can be written in terms of the spectrum of the circle Hamiltonian. Renyi entropies for integer $n$ also establish such a connection for higher genus manifolds. In this sense, it is interesting to note that there is no current understanding of Renyi mutual informations in terms of mathematically well defined quantities in the quantum theory (except for $n=1$). It would be natural to explore integer powers of the modular operator in connection to higher genus partition functions. There is also the interesting case of the Virasoro net for $c>1$ approached above. This does not satisfy strong additivity. A better understanding of the origin of this property is needed for the construction of adequate partition functions, and of Renyi mutual information, in this case.  Another interesting observation is that for $d=2$ CFT's Haag duality for two intervals implies Haag duality for any number of intervals \cite{kawahigashi2001multi}. This should be reflected in that modular invariance for the torus partition function should imply modular invariance of higher genus partition functions and be related to the uniqueness of the partition function in more general topologies and geometries. In addition, notice that Haag duality is not restricted to the vacuum state, and then one may be able to use it for other states representing more general Euclidean geometries.

\textbf{Generalized symmetries from bootstrap data:} We have solved this problem here for some cases with global symmetries. But of course, an interesting problem is to understand the presence of, say, one-form symmetries, related to the violation of Haag duality in a ring, in the bootstrap data. Haag duality for balls instructs us that this should be possible. For non-abelian gauge theories in the lattice, the construction of non-abelian Wilson and 't Hooft loops in terms of local operators was described in \cite{Casini:2020rgj}. It is an important open problem to find the footprints of Wilson and 't Hooft loops in the CFT bootstrap data. Similar questions have been pursued in \cite{Hofman:2024oze}. 

\textbf{Towards a classification of QFT's higher form symmetries:} Regarding the classification of QFT's there are various interesting questions related to the present work. Is any consistent model $\mathcal{T}$ a submodel of a complete one $\mathcal{C}$? We have seen this is the case for finite index in $d=2$, where indeed this completion is non-unique and in general unrelated to a group. For $d>2$ it always exists and is unique for theories with DHR sectors in higher dimensions, this is DHR theorem. This particular type of HDV for orbifolds can be eliminated by extending the neutral theory to the charged one. This extension does not modify the vacuum correlators of the neutral algebra. It does not seem to exist any extension that is additive and complete for higher form symmetries in higher dimensions. Such a completion, that includes particles charged under a gauge theory, would modify the vacuum correlators of the original algebra, and would not be just an extension but a dynamical modification of the theory. Still, it is an important problem to classify the possible fusion categories arising for these type of generalized symmetries.  In that direction see  \cite{Casini:2020rgj} and forthcoming work \cite{ClassGen}.

\textbf{Quantum Gravity:} The necessity of modular invariance has appeared repeatedly in the field of quantum gravity, and more particularly within String Theory, e.g. in the worldsheet Polyakov formulation \cite{Polyakov:1981rd,Friedan:1985ge,Polchinski:1998rq}. In such formulation one not only has a $d=2$ CFT, but a CFT coupled to gravity. These and other aspects bring the necessity of modular invariance. More recently, e.g. \cite{Montero:2016tif,Benini:2022hzx,Heckman:2024obe}, modular invariance has been used in relation with the weak gravity conjecture \cite{Harlow:2022ich} and other conjectures in quantum gravity such as the absence of generalized symmetries \cite{Polchinski:2003bq,Banks:2010zn,Harlow:2018tng,Review}. Our general results relating modular invariance and completeness might then help to better understand the lack of global symmetries in quantum gravity, and potentially provide a unified perspective on several features (e.g. necessary existence of particular brane solutions) of string theory.

\section*{Acknowledgements}
H.C. and J.M thanks useful discussions with G. Torroba, Y. Choi, B. C. Rayhaun, and Y. Zheng.
The work of V.B, H.C and J.M is partially supported by CONICET, CNEA and Universidad Nacional de Cuyo, Argentina.  The work of V.B and J.M  was performed in part at the Aspen Center for Physics during the workshop ``The Microscopic Origin of Black Hole Entropy'', supported by a grant from the Simons Foundation (1161654, Troyer). Y.K. acknowledges financial support of Japan Science and Technology Agency (JST) through
CREST program JPMJCR18T6 and Adopting Sustainable Partnerships for Innovative Research Ecosystem (ASPIRE), Grant Number JPMJAP2318. R.L. acknowledges the Excellence Project 2023–2027 MatMod@TOV awarded to the Department of Mathematics, University of Rome Tor Vergata. J.M acknowledges hospitality and support from the International Institute of Physics, Natal, through Simons Foundation award number 1023171-RC.

\appendix

\section{Universal charged density of states in ``non-invertible symmetry'' scenarios}
\label{app}

In this appendix we show how the derivation of the universal charged density of states for QFT's with finite group symmetries described in \cite{magan2021proof} naturally generalizes to more general scenarios, such as the ones appearing in QFT's in $d=2$. This derivation does not assume conformal invariance.

The derivation uses a relation between relative entropies of quantum complementary algebras unraveled in \cite{Casini:2019kex,Magan:2020ake,Casini:2020rgj}. A proof for the type III algebras was given in \cite{hollands2020variational} (see also \cite{Xu:2018uxc}), following \cite{longo2018relative}.  To explain this relation consider an inclusion of algebras $\mathcal{N}\subset\mathcal{M}$. In this scenario there is a space of conditional expectations $\varepsilon : \mathcal{M}\rightarrow \mathcal{N}$, which are completely positive linear maps from the large algebra to the small algebra satisfying the bimodule property
\be
\hspace{-1mm} \varepsilon\left(n_{1}\,m\,n_{2}\right)=n_{1}\varepsilon\left(m\right)n_{2}\,,\hspace{3mm} \forall m\in\mathcal{M},\,\forall n_{1},n_{2}\in\mathcal{N}.\label{ce_def_prop}
\ee
See \cite{L11,Magan:2020ake} and references therein for the characterization of such spaces of conditional expectations. Crucially, given a state $\omega_{\mathcal{N}}$ in $\mathcal{N}$, and a conditional expectation $\varepsilon$, we can build a state $\omega_{\mathcal{N}}\circ \varepsilon$ in $\mathcal{M}$ by composition. Also, for a general state $\omega$ in $\mathcal{M}$ we can define the following relative entropy
\be 
S_{\mathcal{M}}\left(\omega\mid\omega\circ\varepsilon\right)\;.
\ee
If $\mathcal{M}=\mathcal{N}\vee \mathcal{A}$, where ${\cal A}$ is a set of ``charged operators'', then this quantity is a refined measure of uncertainty of $\mathcal{A}$. It is refined because it takes into account correlations between $\mathcal{N}$ and $\mathcal{A}$.

Given such inclusion, further canonical structure arises by taking commutants. If $\mathcal{N}\subset\mathcal{M}$, then $\mathcal{M}'\subset\mathcal{N}'$. This is nicely encoded in the following quantum complementarity diagram 
\begin{eqnarray}\label{qcd}
\mathcal{M} & \overset{\varepsilon}{\longrightarrow} & \mathcal{N}\nonumber \\
\updownarrow\prime\! &  & \:\updownarrow\prime\\ \label{ecr_diagr}
\mathcal{M}' & \overset{\varepsilon'}{\longleftarrow} & \mathcal{N}'\,.\nonumber 
\end{eqnarray}
The vertical direction goes between commutant algebras. The horizontal direction restricts to the target subalgebra. In particular there is a dual conditional expectation $\varepsilon '$ mapping $\mathcal{N}'$ to $\mathcal{M}'$. In this scenario the main theorem of Ref. \cite{Magan:2020ake} states that for any globally pure state, and a particular conditional expectation $\varepsilon$, there exists a dual conditional expectation $\varepsilon'$ such that
\be
S_{\mathcal{M}}\left(\omega|\omega\circ\varepsilon\right)+S_{\mathcal{N}'}\left(\omega|\omega\circ\varepsilon'\right)=\log\lambda \,, \label{cerp}
\ee
where $\lambda$ is the Jones (or algebraic) index of the dual conditional expectations \cite{Longo:1994xe}. This relation was dubbed the \emph{entropic certainty principle}, due to its intimate connection with the uncertainty principle in quantum mechanics.\footnote{The certainty principle was first found in QFT's with global symmetries \cite{Casini:2019kex}. It was then proven in generic inclusions of type I algebras in \cite{Magan:2020ake}, and shortly after extended to type III algebras \cite{hollands2020variational}.} The positivity of relative entropy implies that whenever one relative entropy saturates to the maximal value given by $\log \lambda$, the dual relative entropy necessarily vanishes. The vanishing of one relative entropy, say the dual primed one, implies $\omega=\omega\circ\varepsilon'$. We will use this in the proof.

Let's now assume that the dual inclusions are type III algebras.\footnote{For the computation of entanglement entropy in vacuum this will be exactly the case. For the computation of thermal expectation values associated to the thermofield double this will not be the case but we will still be able of using the algebraic machinery.} Then we can use the subfactor machinery developed in \cite{Longo:1994xe}. In particular we have the Jones ladder associated to the original inclusion
\be
\cdots\supset\mathcal{M}_1\supset\mathcal{M}\supset\mathcal{N}\supset\mathcal{N}_1\supset\cdots\;,
\ee
together with its commutant counterpart
\be
\cdots\subset\mathcal{M}_1'\subset\mathcal{M}'\subset\mathcal{N}'\subset\mathcal{N}_1'\subset\cdots\;.
\ee
The ladder is built by adjoining subsequent Jones projections to the given algebras. For example, $\mathcal{M}_1 =\mathcal{M}\vee e_\mathcal{N}$, where $e_\mathcal{N}\in \mathcal{N}'$ is the Jones projection associated with the inclusion $\mathcal{N}\subset\mathcal{M}$, namely the projector into the Hilbert space generated by $\mathcal{N}$ from the vacuum.

Now, given the inclusion $\mathcal{N}\subset \mathcal{M}$ and a cyclic and separating vector for both algebras $\vert \Omega\rangle$, the ``canonical endomorphism'' \cite{L11} is defined by
\be  
\gamma (\mathcal{M})\equiv j_{\mathcal{N}}j_{\mathcal{M}} (\mathcal{M}) \subset \mathcal{M}\;.
\ee
This uses the modular conjugations, e.g. $j_{\mathcal{M}}(m)\equiv J_{\mathcal{M}} \,m \, J_{\mathcal{M}}$ and similarly for $j_{\mathcal{N}}(n)\equiv J_{\mathcal{N}} \,n \,J_{\mathcal{N}}$, associated with each algebra and the vector $|\Omega\rangle$. This is an endomorphism for $\mathcal{M}$ and indeed one has $\gamma (\mathcal{M})=\mathcal{N}_1$. We can use it to define a canonical endomorphism $\rho$ in the subalgebra by restricting the action $\rho (\mathcal{N})\equiv \gamma\vert_{\mathcal{N}}\subset \mathcal{N}$. We also have $\rho (\mathcal{N})= \mathcal{N}_2$, i.e. the canonical endomorphism jumps two steps on the Jones ladder. All canonical endomorphisms (for different vectors) are unitarily equivalent to each other \cite{Longo:1994xe}. If $\mathcal{N}\subset \mathcal{M}$ has  finite index $\lambda$, then the canonical endomorphism can be expressed as a direct sum of irreducible endomorphisms $\rho_r$ of $\mathcal{N}$ as
\be \label{candec}
\rho \simeq \oplus_r N_r\, \rho_{r} \;,
\ee
where $N_r$ are some multiplicities, and where the identity endomorphism appears only once. This means we can find a set of partial isometries $\omega_r^i\in \mathcal{N}$, $i=1,\cdots, N_r$, satisfying
\be 
\omega_r^{i\dagger}\omega_{j}^{j}=\delta_{ij}\delta_{rs}\,, \quad\sum_{r,i}\, \omega_r^i \,\omega_r^{i\dagger}=\mathds{1} \,,\quad \omega_r^i\, \rho_{r}(n) =\rho(n)\,\omega_r^i\,,\quad n\in \mathcal{N}\;,
\ee
and the canonical endomorphism explicitly reads
\be 
\rho(n)=\sum_{r,i}\, \omega_{r}^i\,\rho_r(n)\,\omega_{r}^{i\dagger}\;.
\ee
The isometry $\omega_{\mathds{1}}\equiv\omega$ intertwining the identity representation with the canonical endomorphism is unique for irreducible subfactors, i.e those for which $\mathcal{N}' \cap\mathcal{M}= \mathds{C}$. Then the conditional expectation is also unique and reads $\varepsilon (m)= \omega^{\dagger} \gamma (m) \omega$. Ref. \cite{Longo:1994xe} then shows that this structure allows to reconstruct $\mathcal{M}$ from $\mathcal{N}$ and a further isometry $v\in \mathcal{M}$. This isometry $v$ intertwines the identity representation and the canonical endomorphism $\gamma$ of $\mathcal{M}$.   It  contains all charges of the model. Explicitly, if $\psi_r^i$ are the partial isometries creating the irreducible endomorphisms, i.e. $\psi_r^i n=\rho_r(n)\psi_r^i$, then this isometry can be written as
\be 
v=\sum_{ri}\omega_{r}^i\,\psi_r^i\;.
\ee
Also, all $m\in \mathcal{M}$ can be written as $m=n\, v$ for some $n\in \mathcal{N}$. The range projection of this partial isometry is the Jones projection $vv^{\dagger}=e_{\mathcal{M}'}\in \mathcal{M}$ for the dual inclusion $\mathcal{M}'\subset \mathcal{N}'$. See \cite{Longo:1994xe} for a detailed explanations of all these facts.

A key aspect of the Jones ladder is that inclusions are anti-isomorphic in a zig-zag manner due to the modular conjugations. In particular we have that $\mathcal{N}_1\subset\mathcal{N}$ is anti-isomorphic to $\mathcal{M}'\subset\mathcal{N}'$ since $j_{\mathcal{N}}(\mathcal{N})=\mathcal{N}'$ and $j_{\mathcal{N}}(\mathcal{N}_1)=\mathcal{M}'$. Therefore the canonical endomorphism of $\mathcal{M}'\subset\mathcal{N}'$ decomposes similarly to \ref{candec} as
\be \label{candec1}
\tilde{\rho} \simeq \oplus_r N_r\, \tilde{\rho}_{r} \;,
\ee
where the tilde signals these are all endomorphisms of $\mathcal{N}'$. The multiplicites $N_r$ and the dimensions $d_r=\tilde{d}_r$ are equal for both inclusions.\footnote{See \cite{Longo:1994xe} for the proper definition of dimension of an endomorphism.} Also, this implies that we have an algebra of dual partial isometries $\tilde{\omega}_r^{i}\in \mathcal{N}'$ intertwining $\tilde{\rho}_{r}$ in $\tilde{\rho}$, and the associated algebra of projectors into the different representations belonging to the relative commutant
\be 
P_r^i\equiv \tilde{\omega}_r^{i} \tilde{\omega}_r^{i\dagger}\in  \mathcal{N}'\cap \tilde{\rho}(\mathcal{N}')'\;.
\ee
As shown in \cite{Longo:1994xe} we have that
\be \label{procon}
\varepsilon (e_\mathcal{N})=\varepsilon (vv^{\dagger})=\frac{1}{\lambda}\mathds{1}\,\,\,\,\,\,\,\,\,\,\,\,\,\,\,\,\,\varepsilon' (P_r^i)=\frac{d_r}{\lambda}\;.
\ee
Given this tight algebraic structure let's come back to the quantum complementarity diagram \ref{qcd} and the certainty relation \ref{cerp}. In this scenario we will look for situations in which the expectation value for the master charge $v$ creating the canonical endomorphism and belonging to $\mathcal{M}$ is maximal. The idea is to bound the associated relative entropy from below  by focusing on the abelian algebra $\mathcal{P}_v$ generated by the projector $p_v \equiv vv^{\dagger}$ and the orthogonal one $\mathds{1}-p_v$. We have
\be
S_{\mathcal{P}_v}\left(\omega|\omega\circ\varepsilon\right)\leq S_{\mathcal{M}}\left(\omega|\omega\circ\varepsilon\right)\leq\log \lambda\;.
\ee
The lower bound comes from monotonicity of relative entropy. The upper bound follows from positivity of relative entropy and the certainty principle, or by the definition of the index in terms of the Pimsner-Popa bound \cite{Longo:1994xe}. If we are able to find $v$ such that $S_{\mathcal{P}_v}\left(\omega|\omega\circ\varepsilon\right)=\log \lambda$, then we have $S_{\mathcal{M}}\left(\omega|\omega\circ\varepsilon\right)=\log \lambda$ and also $S_{\mathcal{M}'}\left(\omega|\omega\circ\varepsilon'\right)=0$. But the later implies $\omega =\omega\circ\varepsilon'$ and therefore
\be 
\omega (P_r^i) =\omega\circ\varepsilon' (P_r^i)=\frac{d_r}{\lambda}\;.
\ee
This in turn implies that the probability of a given representation is 
\be \label{prg2}
p_r=\sum_i \omega (P_r^i)= \frac{N_r d_r}{\lambda}=\frac{N_r d_r}{\sum_r N_r d_r}\;.
\ee
We can apply this structure in two relevant scenarios. They concern the computation of entropic order parameters and symmetry resolved entropies in vacuum and in the thermofield double. Let's start with entanglement entropy in vacuum. We follow \cite{Casini:2019kex}, namely we define entanglement entropy of region $R_1$ through mutual information between two regions $R_1$ and $R_2$ separated by a shell-like region $R'\equiv (R_1\cup R_2)'$ of small size $\epsilon$. We can compute two mutual informations, namely the one in the complete theory and the one in the incomplete theory. These are different because the complete theory includes non-trivial correlations between charged operators located in $R_1$ and $R_2$. More precisely, since the vacuum is invariant under the global conditional expectation, the conditional expectation property for relative entropies gives the relation \cite{longo2018relative,Casini:2019kex}
\be 
I_{\mathcal{C}}(R_1,R_2)=I_{\mathcal{T}}(R_1,R_2)+S_{\hat{\mathcal{T}}_{R}}\left(\omega|\omega\circ\varepsilon\right)\;,
\ee
where the relative entropy is the entropic order parameter of the disconnected region in the incomplete theory. We remind that the incomplete theory has violations of duality in disconnected regions and their complements. This means that
\be
\mathcal{T}(R)\subset \hat{\mathcal{T}}(R)\;.
\ee
The conditional expectation $\varepsilon$ maps the maximal algebra $\hat{\mathcal{T}_{R}}$ to the additive one in this scenario. The associated complementarity diagram is
\newpage
\bea
\hat{\mathcal{T}}(R) & \overset{\varepsilon}{\longrightarrow} &\mathcal{T}(R)\nonumber \\
\updownarrow\prime \!\! &  & \,\updownarrow\prime\\
\mathcal{T}(R')& \overset{\varepsilon'}{\longleftarrow} & \hat{\mathcal{T}}(R')'\,.\nonumber 
\eea
 To compute the entropic order parameter we focus on the algebras of projectors $\mathcal{P}_{v_{R_1}}$ and $\mathcal{P}_{v_{R_2}}$, generated by  $P_{v_{R_1}}\equiv v_{R_1}v_{R_1}^{\dagger}\in \mathcal{C} (R_1)$ and $P_{v_{R_2}}\equiv v_{R_2}v_{R_2}^{\dagger}\in \mathcal{C} (R_2)$ respectively. We can maximize the correlation by thinking in the limit where the two regions touch each other. Then the modular reflection is geometrical and it is natural to choose
\be 
v_{R_1}= J v_{R_2} J\;.
\ee
The expectation value of the projectors is then
\be 
\omega (P_{v_{R_1}} P_{v_{R_2}})=\omega (P_{v_{R_1}} J P_{v_{R_1}} J)=\omega (P_{v_{R_1}} J P_{v_{R_1}} J)=\omega (P_{v_{R_1}} \Delta^{1/2} P_{v_{R_1}} )\;,
\ee
where $\Delta$ is the modular operator associated with the vacuum and the region $R_1$. Ref. \cite{Casini:2019kex} explains why charged operators can be chosen so that they commute with the modular operator in the limit. Indeed this can only be done
in an approximate sense, choosing the charged operators supported in modes of the
modular Hamiltonian with modular energy tending to zero. A proof can be found in \cite{hollands2020variational,longo2018relative}. For such modes we then have
\be 
\omega (P_{v_{R_1}} P_{v_{R_2}})=\omega (P_{v_{R_1}} \Delta^{1/2} P_{v_{R_1}} )\simeq \omega (\Delta^{1/2} P_{v_{R_1}}  P_{v_{R_1}} )=\omega ( P_{v_{R_1}}  )=\omega\circ\varepsilon ( P_{v_{R_1}}  )=\frac{1}{\lambda}\;.
\ee
This implies the entropic order parameter saturates to its maximum $S_{\hat{\mathcal{T}}_{R}}\left(\omega|\omega\circ\varepsilon\right)=\log\lambda$ and the entropic disorder vanishes $S_{\hat{\mathcal{T}}_{R'}}\left(\omega|\omega\circ\varepsilon'\right)=0$. A more transparent explanation of the vanishing of the entropic disorder parameter comes from the fact that in the limit in which the two regions touch each other, the complementary region pinches. In that limit all operators in the complementary region have zero expectation values, except for the identity of course. From both arguments we arrive at the same conclusion: the probabilities of the different sectors is \ref{prg} and the symmetry resolved entropy in these general scenarios is $S_r=S+\log p_r$. Equivalently, the density matrix is a direct sum of blocks associated with the different projectors, each block with a classical probability given by $p_r$.

A very similar story goes along for the expectation values of projectors in the thermal ensemble. We start from the thermofield double state, defined in the tensor product of two QFT's as
\begin{equation}
\vert\textrm{TFD}\rangle =Z(\beta)^{-1/2}\sum\limits_{i}e^{-\beta E_{i}/2}\vert E_{i}^{R},E_{i}^{L}\rangle\,,
\end{equation}
where $L,R$ corresponds to the Left/Right QFT and play the role of region $R_1$ and $R_2$ in the previous scenario. Thermal expectation values of projector operators acting on the right (or left) system can be written as expectation values in the thermofield double state
\be 
p_r^i= \langle P_r^i\rangle_\beta=\langle\textrm{TFD}\vert \,P_r^i\, \vert\textrm{TFD}\rangle\;.
\ee
The only difference with the previous scenario is that the complementarity diagram changes to
\bea
\mathcal{T}_{L}\vee\mathcal{T}_{R}\vee p_v & \overset{\varepsilon}{\longrightarrow} &\mathcal{T}_{L}\vee\mathcal{T}_{R}\nonumber \\
\updownarrow\prime \!\! &  & \,\updownarrow\prime\\
1& \overset{\varepsilon'}{\longleftarrow} & P_r^i\,.\nonumber 
\eea
This diagram is saying something very simple. In the upper right we have the algebra of the incomplete theory. Since it covers the full spacetime it is a type I algebra. Also, although there is not geometrical complementary region, this algebra can have a center. This center, which is the commutant of the incomplete algebra, contains the projectors into the different superselection sectors, namely the  $P_r^i$. The dual conditional expectation $\varepsilon'$ then kills this dual abelian algebra of projectors, mapping it to the identity. Equivalently, the conditional expectation gives a  particular state on the abelian projector algebra. This $\varepsilon'$ is defined as before \ref{procon}. The commutant of the identity is the full operator algebra, which is the incomplete one plus the Jones projection.

To derive the desired relation at high temperature,  we again need to choose the master charged operators at each side of the thermofield double to be the modular reflected ones. In this case the expectation value of the correlator of Jones projections at each side of the TFD, denoted by $\omega$ below, is 
\be 
\omega (P_{v_{L}} P_{v_{R}})=\omega (P_{v_{L}} J P_{v_{L}} J)=\omega (P_{v_{L}} J P_{v_{L}} J)=\omega (P_{v_{L}} \Delta^{1/2} P_{v_{L}} )\;,
\ee
where $\Delta=e^{-\beta H}$ is the modular operator associated with the TFD. In this case, the high temperature limit takes $\Delta$ to one. We obtain
\be 
\omega (P_{v_{R_1}} P_{v_{R_2}})\xrightarrow[\beta\rightarrow 0]{}  \omega (P_{v_{L}} )=\omega\circ\varepsilon ( P_{v_{L}}  )=\frac{1}{\lambda}\;.
\ee
The understanding of this limit then generalizes the theorem in \cite{magan2021proof}, giving the following equivalent results, holding for a  compact spatial manifold in the high temperature limit:

\begin{itemize}

\item The average density of states of a certain representation $r$ at high energies has the following form
\be
\rho_r(E)=\frac{N_r d_r}{\lambda}\,\rho(E)\,,
\ee  
where $\rho(E)$ is the density of states at energy $E$,  $\lambda=\sum_r N_r d_r$ and $N_r$ are the multiplicities appearing in the canonical endomorphism associated to the subfactor.

\item The expectation values of projectors are
\be
\lim_{\beta\rightarrow 0} \langle P_r^i \rangle= \frac{d_r}{\lambda}\;.
\ee

\item The thermal state has the following decomposition
\be 
\rho_{\beta}\xrightarrow[\beta\rightarrow 0]{}  \oplus_r\,\frac{N_r d_r}{\lambda} \,\rho_{\beta}^r\;.
\ee

\item The thermal entropy has the following decomposition
\be 
S(\rho_{\beta}) =\sum_r  \frac{N_r d_r}{\lambda}\, S(\rho_{\beta}^r)+S(p_r)\;,
\ee
where 
\be 
S(\rho_{\beta}^r)=S(\rho_{\beta})+\log \frac{N_r d_r}{\lambda}\;.
\ee
The quantity $S(\rho_{\beta}^r)$ is called the symmetry resolved entropy.

\item The entropic order parameter in the TFD saturates to its maximum value
\be 
\lim_{\beta\rightarrow 0} S_{\mathcal{T}_{L}\vee\mathcal{T}_{R}\vee p_v}(\omega|\omega\circ \varepsilon)=\log \lambda\;.
\ee

\item The entropic disorder parameter in the TFD vanishes
\be 
\lim_{\beta\rightarrow 0} S_{(\mathcal{T}_{L}\vee\mathcal{T}_{R})')}(\omega|\omega\circ \varepsilon')=0\;.
\ee

\end{itemize}

\bibliographystyle{utphys}
\bibliography{EE}

\end{document}